\documentclass[journal,12pt,onecolumn]{IEEEtran}
\usepackage{graphics}
\usepackage{graphicx}
\usepackage{epsfig}
\usepackage{amsmath,amssymb}
\usepackage{mathrsfs,amsfonts,fancyhdr}
\usepackage{multirow}
\usepackage{multicol}
\usepackage{cite}
\usepackage{amsthm}
\newtheorem{thm}{Theorem}[]
\newtheorem{cor}{Corollary}[]
\newtheorem{lem}{Lemma}[]
\newtheorem{mydef}{Definition}[]

\theoremstyle{remark}
\newtheorem{rem}{Remark}[]

\begin{document}

\title{{\LARGE Spectral Efficiency and Outage Performance for Hybrid D$2$D-Infrastructure Uplink Cooperation}}

\author{\IEEEauthorblockN{Ahmad Abu Al Haija and Mai Vu}
\thanks{Ahmad Abu Al Haija is with the department of Electrical and
  Computer Engineering, McGill University, Montreal, Canada
(e-mail: ahmad.abualhaija@mail.mcgill.ca).This work is performed while he
  visits Tufts University.

  Mai Vu is with the department of Electrical and
  Computer Engineering, Tufts University, Medford, MA, USA (e-mail:
  maivu@ece.tufts.edu).}}

\maketitle
\vspace{-15mm}
\begin{abstract}
We propose a time-division uplink transmission scheme that is applicable to future cellular systems by introducing  hybrid device-to-device (D$2$D)
and infrastructure cooperation. We analyze its spectral efficiency and outage performance and show that compared to existing frequency-division schemes,  the proposed scheme achieves the same or better  spectral efficiency and outage performance while having simpler signaling
and shorter decoding delay. Using time-division, the proposed scheme divides  each transmission frame into  three phases with variable
durations. The two user equipments (UEs)  partially exchange their information in the
first two phases, then cooperatively transmit to the base station (BS) in the
third phase.  We further formulate its common and individual outage probabilities, taking into account outages at both UEs
and the BS. We analyze this outage performance in Rayleigh fading environment assuming full channel state information (CSI) at the
receivers and limited CSI at the transmitters. Results show that comparing to non-cooperative transmission,
the proposed cooperation always improves the instantaneous
achievable rate region even under half-duplex transmission. Moreover, as the received
signal-to-noise ratio increases, this uplink cooperation significantly reduces overall outage probabilities and achieves the full diversity order  in spite of additional
outages at the UEs. These characteristics of the proposed uplink cooperation make
 it appealing for deployment in future cellular networks.

\end{abstract}

\IEEEpeerreviewmaketitle

\section{Introduction}\label{sec:intro}
The escalating growth of wireless networks accompanied with their multimedia services motivates system designers to deploy new technologies that efficiently utilize the wireless spectrum. Since efficiency per link has been approaching  the theoretical limit  for legacy cellular network standards including $2^{\text{nd}}$ and $3^{\text{rd}}$ generations ($2$G and $3$G) \cite{ANJV}, many advanced techniques are proposed  for next generation wireless network standards, Long Term Evolution Advanced  (LTE) and LTE-advance (LTE-A), to improve the spectral efficiency of cellular networks. These techniques  include multi-cell processing \cite{mcell}, heterogeneous network deployment and device-to-device (D$2$D) communication \cite{ANJV}.

In multi-cell processing \cite{mcell}, base stations (BSs) of different cells utilize the backhaul network connecting them to exchange the channel state information (CSI) or the data of their users. Then, they can perform interference coordination or MIMO cooperation (as in coordinated multipoint (CoMP) transmission) to improve their downlink transmissions. For the uplink transmission, it is of interest to study the cooperation among the user equipments (UEs) by utilizing the D$2$D mode in addition to the infrastructure mode. D$2$D communication allows two close UEs to perform direct communication \cite{DTD}. Such D$2$D communication has many applications including  cellular offloading \cite{D2DS3}, video dissemination \cite{D2DS6} and smart city applications \cite{D2DS2}. However, few works have considered hybrid D$2$D and infrastructure cooperation. The main reason is that such  hybrid cooperative transmission is still not
mature enough for inclusion in  specific standards for practical implementation \cite{DTD}. This paper considers uplink user cooperation,
in which two UEs cooperatively transmit to a BS, and analyzes  its spectral efficiency and outage performance.
\begin{figure}[t]
    \begin{center}
    \includegraphics[width=67mm]{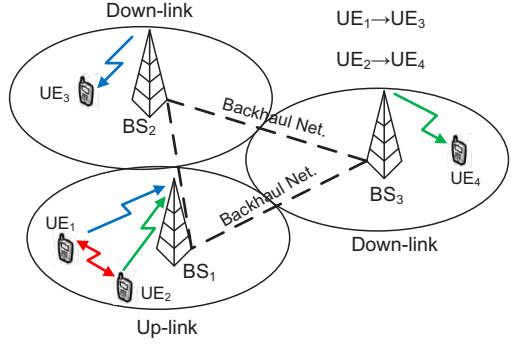}
    \caption{D$2$D cooperation in uplink communication.} \label{fig:system_model}
    \end{center}
\vspace*{-7mm}
\end{figure}
\subsection{A Motivating Example}
Consider the cellular network shown in Figure \ref{fig:system_model} where user equipment $1$ ($\text{UE}_1$) wishes to communicate with $\text{UE}_3$ and $\text{UE}_2$ wishes to communicate with $\text{UE}_4$. This example is valid for both homogeneous and heterogeneous networks as $\text{UE}_1,$  $\text{UE}_2$) and $\text{BS}_1$ can belong to a macro or femto cell.

In the current cellular networks and LTE-A standards \cite{D2DN}, resource partitioning (RP) is used where each user equipment (UE) is given a resource block for its transmission to the base station (BS). The resource blocks are orthogonal in order to reduce the interference as shown in Figure \ref{fig:LTEMOD}. While orthogonal transmission simplifies the signal design at UEs and the decoding at the BS, it poorly utilizes the available spectrum which limits the achievable throughput. Here, the proximity between $\text{UE}_1$ and $\text{UE}_2$ may lead to  strong channel links between the UEs. Hence, adding a D$2$D phase appears as a valuable technique for the two UEs to cooperate in order to improve their throughput to $\text{BS}_1$. Different from pure D$2$D where one UE aims to send information to another UE, in hybrid cooperative transmission, the two UEs have different final destinations but choose to cooperate to help each other send information to the BS.

 Instead of resource partitioning, these UEs can cooperatively transmit to the BS in the same resource blocks, provided that they have exchanged their information beforehand. Such cooperation can be carried out with advanced signal processing at the UEs and/or BS and can significantly improve the spectral efficiency  and outage performance, even when the resource blocks spent for information exchange between the UEs are taken into account.
Existing results have shown that spectral efficiency can be  improved with concurrent transmission where $\text{UE}_1$ and $\text{UE}_2$ transmit concurrently using  the whole
spectrum and the BS decodes using successive interference cancelation (SIC) as in the multiple access channel (MAC) \cite{hansen1977mfr}.
The spectral efficiency can be further improved when $\text{UE}_1$ and $\text{UE}_2$ cooperate to send their information to the BS  by exchanging their information and perform coherent transmission (beamforming) to the BS. Such cooperative transmission requires
advanced processing  at the UEs and the BS as rate splitting and superposition coding are required at the UEs while joint decoding is required at
the BS \cite{Erkip, SErkip}. Thanks to modern computational capability, these advanced processing now appears feasible for upcoming cellular systems. In this paper, we will show that such hybrid D$2$D-infrastructure cooperation can improve not only the spectral efficiency but also the reliability performance in wireless fading channels.

%
\subsection{Literature Review}
%
In \cite{fcc2005fof}, a cooperative channel is first modeled as a multiple access channel with generalized
feedback (MAC-GF) and a full-duplex information-theoretic coding scheme is proposed. This scheme has block Markov signaling where the transmit information
in two consecutive blocks is correlated and employs backward decoding where the BS starts decoding from the last block.
This scheme is adapted to half-duplex transmission using code-division multiple-access (CDMA) \cite{Erkip},  FDMA  \cite{Mes7} and orthogonal FDMA (OFDMA) \cite{SErkip}. The schemes in \cite{Erkip, SErkip} have long delay because of
backward decoding while the scheme in \cite{Mes7} has one block delay because of sliding window decoding but has a smaller rate region because of the specific implementation. In the CDMA scheme,  generating orthogonal codes becomes more complicated for a large number of users.

Whereas existing works on cooperative transmission have been focusing on a frequency-division (FD) implementation \cite{SErkip, Mes7}, this
 paper analyzes a time-division (TD) alternative and show that the TD implementation can achieve the same or better spectral efficiency as the FD
 implementation while having simpler transmit signals  and shorter decoding delay. In this paper, similar to \cite{Erkip, SErkip, Mes7}, we consider two UEs as a basic unit
 of cooperation. It should be noted, however, that extension to $m$ UEs in the uplink transmission is also possible where group of
 UE pairs or all UEs cooperate to send their information.

Further to analyzing the spectral efficiency, we also analyze the outage performance of the cooperative scheme. Outage performance has not been considered in the literature for either cooperative TD or FD implementation. So far, outage has only been considered for the non-cooperative settings. For the non-cooperative MAC, there exist individual and common outages as defined in \cite{outand}. Assuming CSI at the transmitters, the optimal power allocations are derived to minimize the outage capacity. In \cite{NaR}, closed form expressions are derived for the common and individual outages of the two-user MAC assuming no CSI at the transmitters. The
diversity gain region is defined  in \cite{DGRBC} and derived for the  MIMO fading broadcast channel and the MAC using error exponent analysis
in \cite{EEMAC}. The outage probabilities for different relaying techniques in the relay channel have been studied in \cite{Lantse, alouR}. No results so far exist, however, on outage for cooperative multiple access transmission.
In this paper, we analyze the outage performance of the TD cooperative
 transmissions. Since there is no outage performance available for exiting cooperative FD schemes, we also extend our analysis to these schemes
 in order to compare the outage performance.
\subsection{Main Results and Contribution}
In this paper, we propose a TD cooperative hybrid D$2$D-infrastructure transmission
scheme for uplink multiple access communication that can be applied in future cellular systems, derive its achievable rate region and analyze its outage performance over Rayleigh fading channel.  Comparing with the FD schemes
 in \cite{SErkip, Mes7}, the proposed scheme has the same or better rate region, and better outage performance with
simpler signaling and shorter decoding delay. This work is different from our previous work
 \cite{haijatcom}, in which we optimized the power allocation for maximum
 spectral efficiency of a fixed channel, but did not show the ML
 decoding analysis nor consider fading channels and outage analysis. Comparing with our previous scheme in \cite{haivu3}, the proposed scheme achieves the same rate region although it is simpler as it has less splitting for each UE information.

The proposed scheme sends independent information in each transmission block such that the decoding at the end of each block is
possible. To satisfy the
half-duplex constraint, time division is used where each transmission block is divided into $3$
phases. The first two phases are for information exchange between the two UEs, and the last phase is for cooperative transmission to the BS. While the BS is always in the receive mode, the two
UEs alternatively transmit and receive during the first $2$ phases and
coherently transmit during the last phase. The decoding at the BS is performed using
joint maximum-likelihood (ML) receiver among the $3$ phases.

We consider a single antenna at both UEs and the BS but the results can be extended to the MIMO case. We consider block fading channel where all links remain constant over each transmission block and independently vary
in the next block. We assume full  CSI at the receiver side with limited CSI at
the transmitter side where, as in \cite{Erkip}, each UE knows the phase
of its channel to the BS such that the two UEs can employ
coherent transmission. Moreover, each UE knows the relative order between its cooperative and direct links which enables it to cooperate
 when its cooperative link is stronger than the direct link.

 We formulate and analyze both common and individual outages and extend the results in \cite{myglob1} by comparing with the existing RP and frequency division schemes. The individual outage pertains to  incorrect decoding of
 one user information regardless of the other user
 information, while the common outage pertains to incorrect decoding  of
 either user information or both. Because of the information exchanging phases
 transmission, the outage analysis must also consider the outages at the
 UEs. The rate splitting and superposition coding structure also complicates outage analysis and requires dependent analysis of the outage for different information parts. We further derive the outage probabilities for existing FD implementations in \cite{Mes7, SErkip} and compare with our TD implementation.
  Results show that as the received SNR increases, the proposed TD cooperation
 improves outage performance over both orthogonal RP and concurrent non-cooperative transmission schemes in spite of additional outages at the UEs. To the best of our knowledge, this is the
 first work that formulates and analyzes the outage performance for cooperative transmission with rate splitting.
 \subsection{Paper Outline}
 The rest of this paper is organized as follows. Section \ref{sec:system_model} describes the channel model. Section \ref{sec: simpsch} describes the
 proposed time-division cooperative transmission and shows its achievable rate region and the outer bound. Section \ref{sec: outana} formulates and analyzes the common and
 individual outage probabilities of the proposed scheme. Section \ref{sec: regsoutb} formulates the outage probability of existing frequency-division implementations  and compares
 their performance. Section \ref{sec:Num} presents numerical results for the rate region and outage performance of the proposed scheme and existing
 ones. Section \ref{sec:conclusion} concludes the paper.
\section{Channel Model}\label{sec:system_model}
 Consider the uplink communication in Figure \ref{fig:system_model} where $\text{UE}_1$ and $\text{UE}_2$ wish to send their information to $\text{BS}_1$.  In the current LTE-A standard, $\text{BS}_1$ employs resource partitioning (RP) and gives orthogonal resource blocks to the UEs for interference free transmission as shown in Figure \ref{fig:LTEMOD}. However, when $\text{UE}_1$ and $\text{UE}_2$ cooperate to send their information at higher rates, the channel is quite similar to the user cooperative diversity channel defined in \cite{Erkip}. Hence, $\text{BS}_1$ shall change its resource allocation  to facilitate  cooperation and meet the half-duplex constraint in wireless communication where each UE
 can only be either in transmit or receive mode but not in both for the same time and frequency band.


 The proposed transmission scheme uses time division (TD) to satisfy the half-duplex constraint. Instead of dividing the resource block into $2$ orthogonal phases, $\text{BS}_1$ divides the full resource block of $n$ symbols length into $3$ phases with variable durations  $\alpha_1n,$ $\alpha_2n$ and
$(1-\alpha_1-\alpha_2)n$ as shown in Figure \ref{fig:Gausmod}. While $\text{BS}_1$ is always in receiving mode, each UE either transmits or receives during
the first two phases and both of them transmit during the $3^{\text{rd}}$ phase. We consider a single antenna at each UE and the BS but the scheme
can be extended to the MIMO case. Then,
the discrete-time channel model for the half-duplex uplink transmission can be expressed in each phase as follows.
\noindent
\begin{align}\label{recsig}
\text{phase}\;1:&\;\;Y_{12}=h_{12}X_{11}+Z_{12},\;\;Y_{1}=h_{10}X_{11}+Z_{1},\\
\text{phase}\;2:&\;\;Y_{21}=h_{21}X_{22}+Z_{21},\;\;Y_{2}=h_{20}X_{22}+Z_{2},\;\;
\text{phase}\;3:\;\;Y_{3}=h_{10}X_{13}+h_{20}X_{23}+Z_{3},\nonumber
\end{align}
\noindent where $Y_{ij},\;(i,j)\in\{1,2\}$, is the signal received by the $j^{\text{th}}$ UE
during the $i^\text{th}$ phase;  $Y_k,\;k\in\{1,2,3\}$ is the signal received by $\text{BS}_1$ during the $k^\text{th}$ phase; and all the $Z_l,\;l\in\{12,21,1,2,3\}$, are i.i.d complex Gaussian noises with zero mean and unit variance. $X_{11}$ and $X_{13}$ are the signals transmitted from
  $\text{UE}_1$ during the $1^{\text{st}}$ and $3^{\text{rd}}$ phases, respectively. Similarly, $X_{22}$ and $X_{23}$ are the signals transmitted from $\text{UE}_2$
 during the $2^{\text{nd}}$ and $3^{\text{rd}}$ phases.

 Each link coefficient is affected by Rayleigh fading and path loss as follows.
 \begin{align}
 h=\frac{\tilde{h}}{d^{\gamma/2}}
 \end{align}
where $\tilde{h}$ is the small scale fading component and has a complex Gaussian distribution with zero mean and unit variance $({\cal N}(0,1))$. The large scale fading component is captured
 by a path loss model where $d$ is the distance between two nodes in the network and $\gamma$ is the attenuation factor. Let $g=|h|$ and $\theta$ be
 the  amplitude  and the phase of a link coefficient, then $g$ has Rayleigh distribution while $\theta$ has uniform distribution in the interval
 $[0,2\pi]$.

  We assume receiver knowledge for the channel coefficient, i.e., $\text{BS}_1$ knows $h_{10}$ and $h_{20}$, $\text{UE}_1$ knows $h_{21}$ and $\text{UE}_2$ knows $h_{12}$. We further assume that $\text{BS}_1$ knows $h_{21}$  and $h_{12}$ which can be forwarded to $\text{BS}_1$ by $\text{UE}_1$ and $\text{UE}_2$, respectively. Moreover, each UE knows the phase of its direct link to $\text{BS}_1$ and the relative amplitude order between its cooperative and direct links. This information can be obtained through feedback from $\text{BS}_1$ since it knows all channel coefficients. The phase knowledge allows UEs to perform coherent transmission to $\text{BS}_1$ and utilize
  the advantage of beamforming while the relative amplitude orders helps decide the best transmission scenario as shown in Section \ref{sec: regcases}.

   We assume block fading where the channel coefficients stay constant in each block through all $3$ phases and change independently in the next block.
\noindent
\begin{figure}[t]
    \centering
    \begin{minipage}[b]{0.48\linewidth}
    \includegraphics[width=0.90\textwidth]{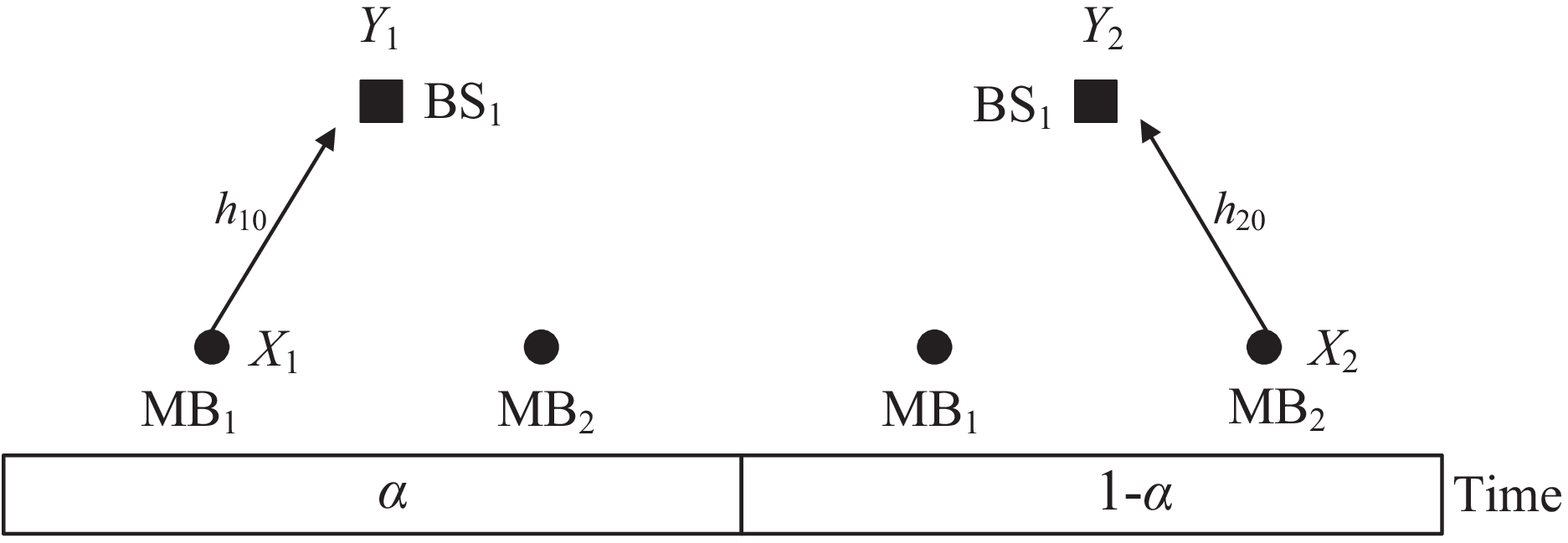}
    \caption{Resource partitioning in current LTE standard.} \label{fig:LTEMOD}
    \end{minipage}
    \hfill
\begin{minipage}[b]{0.48\linewidth}
    \includegraphics[width=0.90\textwidth]{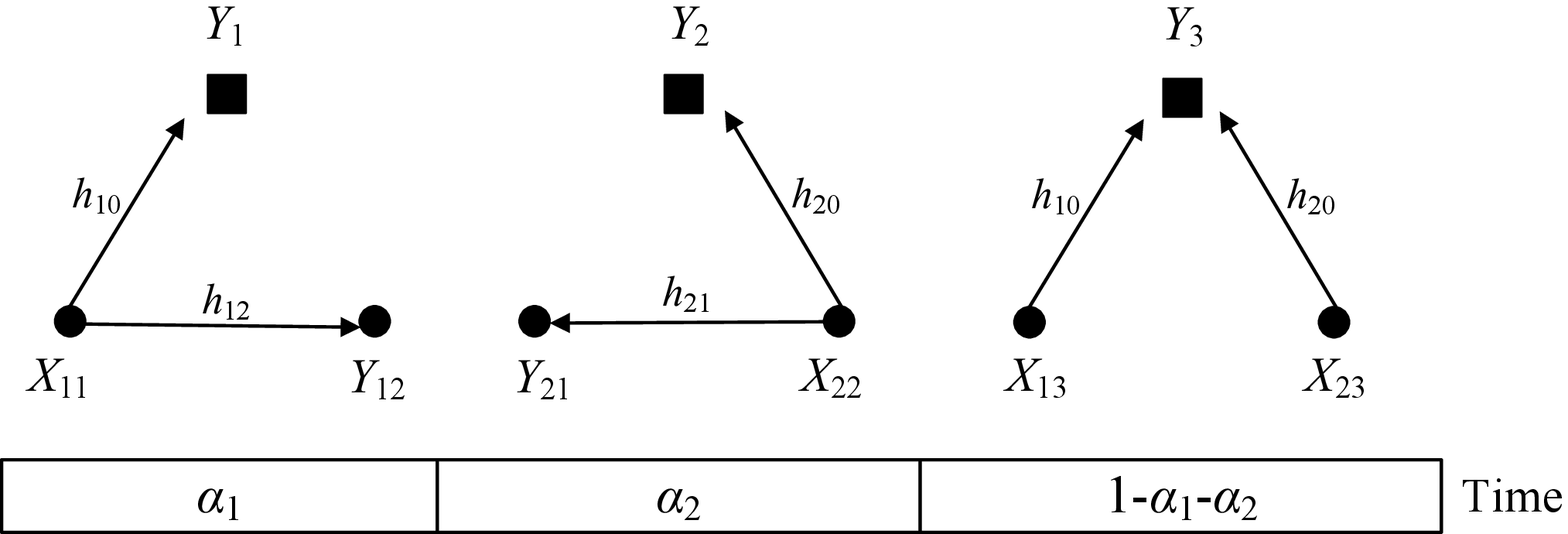}
    \caption{Cooperative uplink transmission.} \label{fig:Gausmod}
    \end{minipage}
\vspace*{-7mm}
\end{figure}
\section{A Time-Division (TD) Uplink Cooperative Device-to-Device Transmission Scheme}\label{sec: simpsch}
Here, we  describe a TD D$2$D cooperative scheme applied to the half-duplex uplink communication in LTE-A networks.  We
also analyze an outer bound and
compare it to the achievable rate region.

Compared with the scheme in \cite{fcc2005fof, Erkip}, the proposed scheme has better spectral efficiency, simpler signaling and shorter decoding delay (no block decoding delay). These characteristics appear since the two UEs transmit among the whole bandwidth, encode independent information in each transmission block and the BS decodes directly at the end of each block instead of backward decoding.  The proposed scheme is based on rate splitting, superposition coding and partial decode-forward (PDF) relaying  techniques.

 The transmission in each block is divided into three phases with relative durations  $\alpha_1,$ $\alpha_2$ and $\alpha_3=1-\alpha_1-\alpha_2$. In each block, $\text{UE}_1$ splits its information
into two parts: a cooperative part with index $i$ and a private part with index $j$. It sends the private part directly to the BS at rate $R_{10}$ and sends the cooperative part to the BS in cooperation with $\text{UE}_2$ at rate $R_{12}$. These parts are then encoded using superposition coding, in which for each transmit sequence of the first information
part, a group of sequences is generated for the second information parts. Similarly, $\text{UE}_2$ splits its information into a cooperative part (indexed by $k)$ and a private part (indexed by $l)$ and encodes them using superposition coding.
 In the first two phases, the two UEs exchange the cooperative information parts. In the $3^{\text{rd}}$ phase, each UE sends both cooperative
 information parts and
 its own private part to the BS. Effectively each UE performs PDF relaying of the cooperative part of
 the other UE. Next, we describe in detail the transmit signaling and ML receiver.
 \subsection{Transmit Signals}
 \subsubsection{Transmit sequences generation}
 As in all communication systems, the channel encoder maps each piece of input information into a unique sequence. This sequence includes some controlled redundancy of the input information which can be used by the receiver to alleviate the noise encountered during transmission to reduce decoding error.

 Let ${\cal I }$ $({\cal K })$ and ${\cal J}$ $({\cal L})$  be the sets of signal indices for the cooperative and private parts of
  $\text{UE}_1$ $(\text{UE}_2)$, respectively.
Since the transmission is affected by Gaussian noise as in (\ref{recsig}), both UEs employ Gaussian signaling to maximize the transmission rate \cite{hsce612}. The Gaussian signals are generated as follows. For each element $i\in{\cal I}$, independently generate a signal vector (sequence)  $\boldsymbol{u}_{1,i}$ of length $n$ according to a Gaussian distribution with zero mean and unit variance. This sequence will be scaled by a power allocated by  $\text{UE}_1$ as shown in Section \ref{sec:sign}.
   Similar Gaussian sequences $\boldsymbol{u}_{2,k}$ and $\boldsymbol{u}_{3,i,k}$ are generated for each element $k\in{\cal K}$ and each pair $(i,k)$, respectively. Next, perform
  superposition signaling where for each sequence $\boldsymbol{u}_{3,i,k}$, generate
  a Gaussian sequence $\boldsymbol{x}_{13,j,i,k}$ $(\boldsymbol{x}_{23,l,i,k})$ for each $j\in{\cal J}$ $(l\in{\cal L})$.
  The superposition coding reduces the decoding complexity and increases the rate region as
 shown in Section \ref{MLD}.
\begin{figure}[t]
    \begin{center}
    \includegraphics[width=85mm]{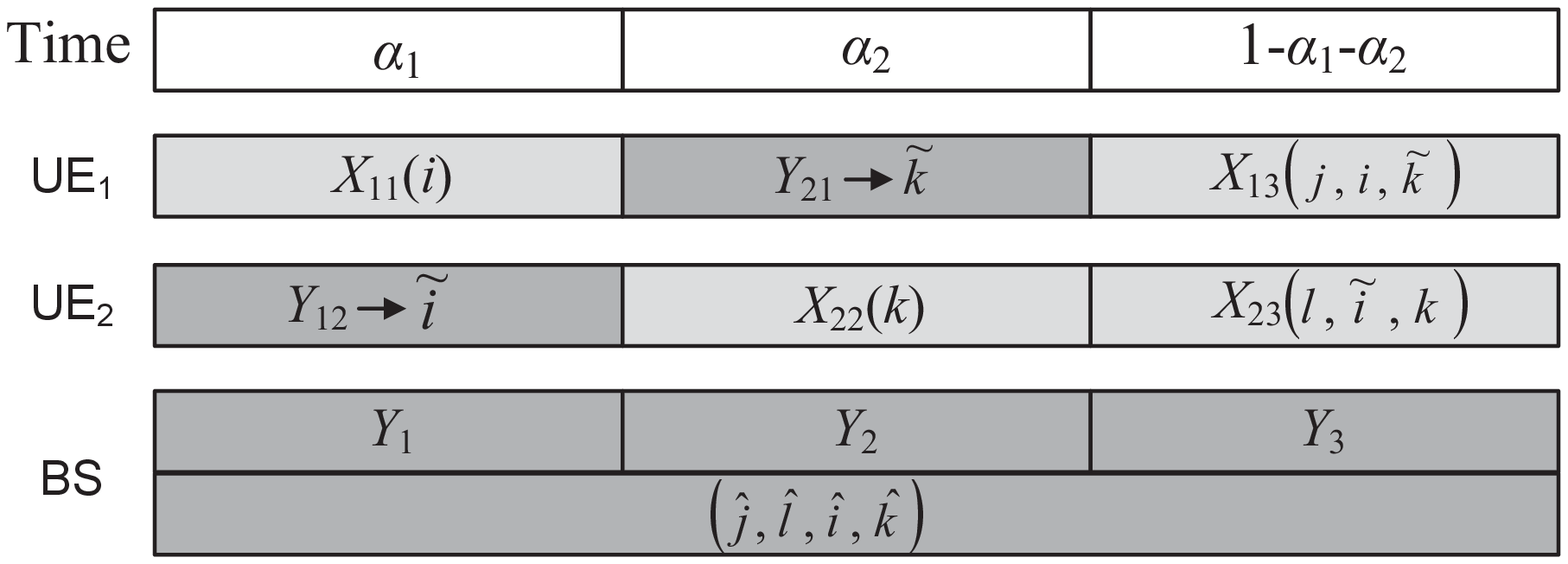}
    \caption{Cooperative uplink transmission scheme (light shade$=$transmit signal, dark shade$=$received signal and decoded indices). } \label{fig:trans}
    \end{center}
 \vspace*{-7mm}
\end{figure}
\subsubsection{Transmission scheme}\label{sec:sign}
 In the $1^{\text{st}}$ phase,  $\text{UE}_1$ sends its cooperative information at rate $R_{12}$ by transmitting the signal $X_{11,i}$ which consists of the first $\alpha_1n$ elements of a scaled sequence of $\boldsymbol{u}_{1,i}$ as shown in (\ref{sigtr}). By the end of the $1^{\text{st}}$ phase,  $\text{UE}_2$ decodes $X_{11,i}$. Then, in the $3^{\text{rd}}$ phase,  $\text{UE}_1$ sends its private information and both cooperative information  at  rate triplet $(R_{10},R_{12},R_{21})$ by transmitting the signal $X_{13,j,i,k},$ which consists of the last $\alpha_3n$ elements of sequence $\boldsymbol{x}_{13,j,i,k}$. Similarly,  $\text{UE}_2$ transmit the signals $X_{22,k}$ and $X_{23,l,i,k}$ in the $2^{\text{nd}}$and $3^{\text{rd}}$, respectively. Since both UEs
know indices $i$ and $k$ in this phase, they can perform coherent transmission of these cooperative information by transmit beamforming such that the achievable
rates of both UEs are increased.
The transmit signals at each phase are
\noindent
\begin{align}\label{sigtr}
\text{phase}\;1:&\;X_{11,i}=\sqrt{\rho_{11}}U_{1}(i),\;\;\;
\text{phase}\;2:\;X_{22,k}=\sqrt{\rho_{22}}U_{2}(k),\\
\text{phase}\;3:&\;X_{13,j,i,k}=\sqrt{\rho_{10}}V_{1}(j)+\sqrt{\rho_{13}}U_3(i,k),\;\;
X_{23,l,i,k}=\sqrt{\rho_{20}}V_{2}(l)+\sqrt{\rho_{23}}U_3(i,k)\nonumber
\end{align}
\noindent where $U_{1},U_{2},V_{1},V_{2}$ and $U_3$ are independent and identically distributed Gaussian signals with zero mean and unit variance, $X_{13}$ and $X_{23}$ are superpositioned in $U_3$. Here, $\rho_{11},\rho_{22},\rho_{10}$ and $\rho_{20}$ are the transmission powers allocated for signals $U_{1},U_{2},V_{1}$ and $V_{2}$, respectively,
$\rho_{13}$ and $\rho_{23}$ are the transmission powers allocated for signal $U_3$ by  $\text{UE}_1$ and  $\text{UE}_2$, respectively. Let ${\cal P}_1$ and ${\cal P}_2$ be the total transmission power for  $\text{UE}_1$ and  $\text{UE}_2$,
respectively. Then, we have the following  power constraints:
\noindent
\begin{align}\label{powcsch2}
\alpha_1\rho_{11}+\alpha_3(\rho_{10}+\rho_{13})&={\cal P}_1,\;\;\;
\alpha_2\rho_{22}+\alpha_3(\rho_{20}+\rho_{23})={\cal P}_2.
\end{align}
\subsection{ML receiver}\label{MLD}
 Assume that all sequences in any set ${\cal I},$ ${\cal J},$ ${\cal K},$ or ${\cal L},$ have equal transmission probability. Sequence maximum likelihood (ML) criterion is then optimal and achieves the same performance as maximum a posterior probability (MAP) criterion.
\subsubsection*{At each UE}
In the $1^{\text{st}}$ phase,  $\text{UE}_2$ detects $i$ from $Y_{12}$ using sequence maximum likelihood (ML) criterion. Hence, for a given sequence $\boldsymbol{y}_{12}$ of length $\alpha_1n$,  $\text{UE}_2$ chooses $\boldsymbol{\hat{x}}_{11,\hat{i}}$ to be the transmitted sequence if
\begin{align}\label{decr1}
p(\boldsymbol{y}_{12}|\boldsymbol{\hat{x}}_{11,\hat{i}})\geq p(\boldsymbol{y}_{12}|\boldsymbol{x}_{11,i}),\;\;\text{for all}\;\;\boldsymbol{x}_{11}(i)\neq\boldsymbol{\hat{x}}_{11}(\hat{i})
\end{align}
  $\text{UE}_1$ applies similar decoding rule in the $2^{\text{nd}}$ phase. Hence,  $\text{UE}_1$ and  $\text{UE}_2$ can reliably detect the
  transmit sequences $\boldsymbol{x}_{11,i}$ and $\boldsymbol{x}_{22,k}$, respectively, if
\begin{align}\label{spr12}
R_{12}\leq&\; \alpha_1\log\left(1+g_{12}^2\rho_{11}\right)=J_1\;\;
\text{and}\;\; R_{21}\leq\;\alpha_2\log\left(1+g_{21}^2\rho_{22}\right)=J_2.
\end{align}
\subsubsection*{At the base station}
The BS utilizes the received signals in all three phases $(Y_1,Y_2,Y_3)$ to
jointly detect all information parts $(j,l,i,k)$  using joint sequence ML criterion. With the signaling in (\ref{sigtr}), the received signals at the BS are given as follows.
\begin{align}\label{bsml}
\text{Phase 1:}&\;Y_1=h_{10}\sqrt{\rho_{11}}U_1(i)+Z_1,\;\;\text{Phase 2:}\;Y_2=h_{20}\sqrt{\rho_{11}}U_2(k)+Z_2,\nonumber\\
\text{Phase 3:}&\;Y_3=h_{10}\sqrt{\rho_{10}}V_1(j)+h_{20}\sqrt{\rho_{20}}V_2(l)+\left(h_{10}\sqrt{\rho_{13}}+h_{20}\sqrt{\rho_{23}}\right)U_3(i,k)+Z_3,
\end{align}
Then,  for given received sequences $\boldsymbol{y}_1$ of length $\alpha_1n,$ $\boldsymbol{y}_2$ of length $\alpha_2n$ and $\boldsymbol{y}_3$ of length $\alpha_3n,$ the BS chooses $\hat{\boldsymbol{x}}_{11,\hat{i}},$ $\hat{\boldsymbol{x}}_{22,\hat{k}},$ $\hat{\boldsymbol{x}}_{10,\hat{j},\hat{i},\hat{k}}$
and $\hat{\boldsymbol{x}}_{20,\hat{l},\hat{i},\hat{k}}$   to be the transmitted
sequences if:
\begin{align}\label{decr2}
&P(\boldsymbol{y}_1|\hat{\boldsymbol{x}}_{11,\hat{i}})P(\boldsymbol{y}_2|\hat{\boldsymbol{x}}_{22,\hat{k}})
P(\boldsymbol{y}_3|\hat{\boldsymbol{x}}_{10,\hat{j},\hat{i},\hat{k}},\hat{\boldsymbol{x}}_{20,\hat{l},\hat{i},\hat{k}})\geq
P(\boldsymbol{y}_1|\boldsymbol{x}_{11,i})P(\boldsymbol{y}_2|\boldsymbol{x}_{22,k})P(\boldsymbol{y}_3|\boldsymbol
{x}_{10,j,i,k},\boldsymbol{x}_{20,l,i,k})\nonumber\\
\text{for all}&\;\boldsymbol{x}_{11,i}\neq \hat{\boldsymbol{x}}_{11,\hat{i}},\;
 \boldsymbol{x}_{22,k}\neq \hat{\boldsymbol{x}}_{22,\hat{k}},\;
 \boldsymbol{x}_{10,j,i,k}\neq \hat{\boldsymbol{x}}_{10,\hat{j},\hat{i},\hat{k}}
 \;\;\text{and}\;\;
 \boldsymbol{x}_{20,l,i,k}\neq \hat{\boldsymbol{x}}_{20,\hat{l},\hat{i},\hat{k}}
\end{align}
\begin{lem}
For each channel realization, the rate constraints that ensure vanishing decoding error probabilities at the BS are given as
\begin{align}\label{sprd}
R_{10}\leq&\; \alpha_3\log\left(1+g_{10}^2\rho_{10}\right)=J_3,\;\;\;
R_{20}\leq\; \alpha_3\log\left(1+g_{20}^2\rho_{20}\right)=J_4\\
R_{10}+R_{20}\leq&\;\alpha_3\log\left(1+g_{10}^2\rho_{10}+g_{20}^2\rho_{20}\right)=J_5\nonumber\\
R_1+R_{20}\leq&\; \alpha_1\log\left(1+g_{10}^2\rho_{11}\right)\;+\alpha_3\zeta=J_6,\;\;
R_{10}+R_2\leq\;\alpha_2\log\left(1+g_{20}^2\rho_{22}\right)+\alpha_3\zeta=J_7\nonumber\\
R_1+R_2\leq&\; \alpha_1\log\left(1+g_{10}^2\rho_{11}\right)+\alpha_2\log\left(1+g_{20}^2\rho_{22}\right)+\alpha_3\zeta=J_8,\nonumber\\
\zeta=&\;\log\big(1+g_{10}^2\rho_{10}+g_{20}^2\rho_{20}+(g_{10}\sqrt{\rho_{13}}+g_{20}\sqrt{\rho_{23}})^2\big).\nonumber
\end{align}
\end{lem}
Hence, the BS can reliably decode all information parts if the constraints in (\ref{sprd}) are satisfied.
Note that the terms $J_6, J_7$ and $J_8$ show the advantage of beamforming resulted from coherent transmission of $(i,k)$ from both UEs in the
$3^{\text{rd}}$ phase.
\begin{proof}[Sketch of the proof]
A decoding error can occur for the cooperative or the private parts or both. However, because of superposition coding, if either cooperative part is incorrectly decoded, both private parts will also be decoded incorrectly. Hence, we consider two cases:
\begin{enumerate}
  \item The cooperative parts are decoded correctly:\\
  When both cooperative parts have been decoded correctly, the BS can decode the private parts from $Y_3$ in (\ref{bsml}) after removing $U_3(i,k)$. Then, $Y_3$ becomes similar to the received signal in a MAC. Hence, the rate constraints for the private parts are similar to those of a MAC as given by $J_3,$ $J_4,$ and $J_5$ in (\ref{sprd}).
  \item Either cooperative part or both are decoded incorrectly:\\
  This case contains three sub-cases: only one of the two cooperative parts is decoded correctly, or both are decoded incorrectly, each of which leads to a different rate constraint. If the BS decodes $i$ incorrectly but decodes $k$ correctly, then both $j$ and $l$ will be decoded incorrectly. Because of the joint decoding performed at the BS as in (\ref{decr2}), this incorrect decoding will result in a constraint on the total rate of the parts $i,$ $j$ and $l$. Since
  $i$ is sent in phases $1$ and $3$, this rate constraint is obtained from $Y_1$ and $Y_3$ as follows:
  \begin{align}\label{onecow}
  R_{12}+R_{10}+R_{20}\leq \alpha_1\log\left(1+\text{SNR}_1\right)+\alpha_3\log\left(1+\text{SNR}_3\right),
  \end{align}
  where $\text{SNR}_t,$ $t\in\{1,2,3\}$ is the SNR of all the received signals (due to BS joint decoding) in phase $t$. Thus  (\ref{onecow}) reduces
  to $J_6$ in (\ref{sprd}) where $\text{SNR}_3$ is obtained over all  incorrectly decoded signal parts $(V_1,V_2,U_3)$ in $Y_3$. Similarly, we can
  obtain $J_7$  if the BS decodes $k$ incorrectly but decodes $i$ correctly.  If the BS decodes both $i$ and $k$ incorrectly, then all messages
  parts $i, k, j$ and  $l$ will be in error and we obtain the following rate constraint:
   \begin{align}\label{onetrcow}
  R_{12}+R_{10}+R_{21}+R_{20}\leq \alpha_1\log\left(1+\text{SNR}_1\right)+\alpha_2\log\left(1+\text{SNR}_2\right)+\alpha_3\log\left(1+\text{SNR}_3\right),
  \end{align}
which results in constraint $J_8$ in (\ref{sprd}).
\end{enumerate}
A full analysis based on ML decoding can be found in Appendix A.
\end{proof}
We note that the achievable rate region in (\ref{sprd}) is a direct result of the joint ML decoding performed at the BS simultaneously over all three phases
as in (\ref{decr2}). If the BS uses sequential decoding or decodes each phase separately, this can reduce the decoding complexity but will result
in a strictly smaller rate region.
\vspace{-5mm}
\subsection{Achievable Rate Region and Transmission Scenarios}\label{sec: regcases}
The achievable rate region in terms of $R_1=R_{10}+R_{12}$ and $R_2=R_{20}+R_{21}$ is given as follows.
\vspace{-2 mm}
\begin{thm}\label{nathr1}
The achievable rate region resulting from the proposed scheme for each channel realization consists of rate pairs $(R_1,R_2)$ satisfying the following constraints:
\noindent
\begin{align}\label{th1Grrsimp}
R_1\leq&\; J_1+J_3,\;\;R_2\leq\; J_2+J_4,\;\;
R_1+R_2\leq\; J_1+J_2+J_5,\;\;R_1+R_2\leq\; J_8,
\end{align}
\noindent for some $\alpha_1\geq0,\;\alpha_2\geq 0,$ $\alpha_1+\alpha_2\leq1$ and power allocation set $(\rho_{10},\rho_{20},\rho_{11},\rho_{22},\rho_{13},\rho_{23})$ satisfying (\ref{powcsch2}) where
$J_1$---$J_8$ are given in (\ref{spr12}) and (\ref{sprd}).
\end{thm}
\begin{proof}
Obtained by combining (\ref{spr12}) and (\ref{sprd}). See Appendix A for more details.
\end{proof}
Combining (\ref{spr12}) and (\ref{sprd}) leads to the constraints in (\ref{th1Grrsimp}) in addition to
$2$ other constraints including $R_1+R_2\leq J_1+J_7$ and $R_1+R_2\leq J_2+J_6$. However, these constraints are redundant as stated in the following corollary:
\begin{cor}\label{cor1}
Two sum rate constraints $(R_1+R_2\leq \min\{(J_1+J_7,\;J_2+J_6)\})$ on the achievable region result from combining (\ref{spr12}) and (\ref{sprd}) are redundant.
\end{cor}
\begin{proof}
For any channel configuration, $\min\{(J_1+J_7,\;J_2+J_6)\})\geq \min\{(J_1+J_2+J_5,\;J_8)\})$.
\end{proof}
From the proposed scheme, $4$ optimal sub-schemes can be obtained
depending on the channel configuration. These schemes have different power allocation and phase durations  that are results of the operating scenario. Each UE requires only the relative amplitude order between its cooperative and direct links to determine which scenario to operate. Since the BS knows all links as stated in Section \ref{sec:system_model} and there are $2$ pairs of links, this knowledge can be obtained through a $2$-bit feedback from the BS which incurs negligible overhead. Each bit indicates the relation between one pair of direct and cooperative links. Assume that at the beginning of each transmission block, the UEs have sufficient knowledge of the link orders, the operating scenarios are given as follows.
\subsubsection{Case $1$ $(g_{12}\leq g_{10}$ and $g_{21}\leq g_{20})$, Direct transmissions for both UEs}.

In this case, decoding at the two UEs actually limits the achievable rates because the inter-UE links are weaker than the direct links.
Therefore, both UEs transmit directly to the BS all the time without cooperation as in the concurrent transmission with SIC.  The achievable rate is given in (\ref{th1Grrsimp}) but with $\alpha_1=\alpha_2=0,$ $\rho_{11}=\rho_{13}=\rho_{22}=\rho_{23}=0,$ $\rho_{10}={\cal P}_1$ and
$\rho_{20}={\cal P}_2$.
\subsubsection{Case $2$ $(g_{12}>g_{10}$ and $g_{21}>g_{20})$, Cooperation for both UEs}
In this case, both UEs obtain mutual benefit from cooperation for sending 
 their information to the BS. When $g_{12}>g_{10}$ and $g_{21}>g_{20}$,
 $J_2+J_6>J_8,$ and $J_1+J_7>J_8$. Therefore, the rate constraints are as given in (\ref{th1Grrsimp})  with all signals and phases.
\subsubsection{Case $3$ $(g_{12}>g_{10}$ and $g_{21}\leq g_{20})$, Cooperation for  $\text{UE}_1$ and direct transmission for  $\text{UE}_2$}
Here,  $\text{UE}_1$  prefers cooperation while  $\text{UE}_2$ transmits directly to the BS. Therefore, the transmission is carried over $2$ phases only where $\text{UE}_2$ relays information for $\text{UE}_1$ while also transmitting its own information.  $\text{UE}_1$
sends its cooperative part in the $1^{\text{st}}$ phase. In the $2^{\text{nd}}$ phase,  $\text{UE}_1$ sends its two parts while  $\text{UE}_2$ sends
its full information and the cooperative part of  $\text{UE}_1$. The achievable rate is given in (\ref{th1Grrsimp})
with $\alpha_2=0,$ and $\rho_{22}=0$.
\subsubsection{Case $4$ $(g_{12}\leq g_{10}$ and $g_{21}> g_{20})$, Cooperation for  $\text{UE}_2$ and direct transmission for  $\text{UE}_1$}
This case is the opposite of the Case $3$ where the achievable rate is given in (\ref{th1Grrsimp}) with $\alpha_1=0,$ and $\rho_{11}=0$.
\subsection{Outer Bound}\label{sec:out. cap}
In this section, we provide an outer bound with constraints similar to that in Theorem $1$.
During the $3^{\text{rd}}$ phase, the channel looks like a MAC with common message \cite{haykin2005crb} while during the first two phases, it looks like a broadcast channel (BC). Furthermore, when one UE has no information to send, the channel becomes as the relay channel (RC). Although capacity is known for the MAC with common message and for the Gaussian BC, the capacity for RC
 is unknown in general. In \cite{dbout}, an outer bound is derived for the full-duplex scheme in \cite{Erkip} based on the idea of
 dependence balance \cite{ray4}. When applied to the proposed half-duplex transmission, the outer bound holds without dependence balance
 condition as follows.
 \begin{cor}\label{outcorG}\cite{dbout}
An outer bound for the uplink half-duplex D$2$D communication  consists of all rate pairs $(R_1,R_2)$ satisfying (\ref{th1Grrsimp}) but
replacing $g_{12}^2$ $(g_{21}^2)$ by $g_{12}^2+g_{10}^2$ $(g_{21}^2+g_{20}^2)$ as follows.
\begin{align}\label{outregion}
\!\!\!\!R_1&\leq \alpha_1\log\left(1+(g_{10}^2+g_{12}^2)\rho_{11}\right)+J_3,\;\;R_2\leq \alpha_2\log\left(1+(g_{20}^2+g_{21}^2)\rho_{22}\right)+J_4,\\
\!\!\!\!R_1+R_2&\leq \alpha_1\log\left(1+(g_{10}^2+g_{12}^2)\rho_{11}\right)+\alpha_2\log\left(1+(g_{20}^2+g_{21}^2)\rho_{22}\right)+J_5,\;\;R_1+R_2\leq J_8.\nonumber
\end{align}
\end{cor}
\subsubsection*{MIMO View}
These bounds can also be obtained
using MIMO bounds at receiver and transmitter sides as follows.
%

 Consider the $1^{\text{st}}$ phase,  $\text{UE}_1$ transmits while  $\text{UE}_2$ and the BS receive with full cooperation as in a SIMO $(1\times2)$ channel, this gives the first outer bound on $R_1$  in (\ref{outregion}).  The second bound (on  $R_2$) is obtained in a similar way. For the third bound (on $R_1+R_2$), the $1^{\text{st}}$ and $2^{\text{nd}}$ phases are bounded using a SIMO, similarly to that for $R_1$ and $R_2$, respectively; in the $3^{\text{rd}}$ phase, since we use the SIMO bound at the receiver side, both UEs transmit without cooperation  which results in the term $J_5$.
Finally, the fourth bound (on $R_1+R_2$) is obtained from the MISO bound at the transmitter side: in the $1^{\text{st}}$ phase, only  $\text{UE}_1$ sends and the BS receives given known signal from  $\text{UE}_2$; the same holds for the $2^{\text{nd}}$ phase; in the $3^{\text{rd}}$ phase, both UEs transmit with full cooperation as in a MISO $(2\times 1)$ channel.

Note that the tightness of the outer bound is determined by the ratios $g_{12}^2/g_{10}^2$ and
$g_{21}^2/g_{20}^2$. The outer bound becomes tighter as these two ratios increase since then $g_{12}^2\rightarrow g_{12}^2+g_{10}^2$ and $g_{21}^2\rightarrow g_{21}^2+g_{20}^2$.
In other words, the bound becomes increasingly tight as the inter-UE link qualities increase.
\section{Outage Probability and Outage Rate Region}\label{sec: outana}
The previous analysis provides the region of transmission rates that can be achieved for each fading channel realization.
In most wireless services, however,  a minimum target information rate is required to support the service, below which the service is unsustainable.
For a particular fading realization, the channel may or may not support the target rate. The probability that the rate supported by the fading channel falls below the target rate is called the outage rate probability. Outage has been analyzed for non-cooperation concurrent transmission with SIC (classical MAC) \cite{NaR, outand} but has not been formulated or analyzed in a cooperative setting.

In this section, we formulate and analyze the outage probability of the
 proposed cooperative scheme.
 Suppose that based on the service requirements, the target rate pair is
 $(R_{1},R_{2})$. Outage occurs in the event that the target rate pair lies
 outside the achievable region for a channel realization.
There are two types of outage in multi-user transmission: common and
individual outage \cite{NaR, outand}. The individual outage
for  $\text{UE}_1$ is the probability that the channel cannot support its
transmission rate regardless of whether the channel can or cannot support
the transmission rate of  $\text{UE}_2$. Similar holds for  $\text{UE}_2$. The common
outage is the probability that the channel cannot support the transmission
rate of either  $\text{UE}_1$ or  $\text{UE}_2$ or both.

Unlike the non-cooperative schemes where outage occurs only at the BS, outage in the proposed cooperative scheme can also occur at the
 UEs. Moreover, the outage formulation can be different for each
 channel configuration depending on the specific transmission scheme used
 for that realization as outlined  in the $4$ cases in  Section \ref{sec:
   regcases}.

Define $P_{cm},$ $P_{1m}$ and $P_{2m}$ for $m\in\{1,2,3,4\}$ as the common
and individual outage probabilities for case $m$ as discussed in Section
\ref{sec: regcases}. Then, the outage probability is given as follows.
\begin{thm}\label{outthm}
For the proposed $3$-phase D$2$D uplink scheme, the average common outage probability $(\bar{P}_{c})$
is given as
\begin{align}\label{outtc}
\bar{P}_{c}&=P[g_{12}\leq g_{10}, g_{21}\leq g_{20}]P_{c1}+P[g_{12}> g_{10}, g_{21}> g_{20}]P_{c2}\nonumber\\
&\;\;+P[g_{12}> g_{10}, g_{21}\leq g_{20}]P_{c3}+P[g_{12}\leq g_{10}, g_{21}> g_{20}]P_{c4}.
\end{align}
where $P_{c1}$ and $P_{c4}$ are explained in Sections \ref{seccas1} and \ref{seccas34}, respectively, $P_{c2},$ and $P_{c3}$ are given in
(\ref{outc2}) and (\ref{outc3}), respectively.
The  average individual outage probabilities $(\bar{P}_{1},\bar{P}_2)$ have
similar formulation.
\end{thm}
\begin{proof}
Obtained by formulating the outage probability of each case as in the following sections. 
\end{proof}
\subsection{Outage Probability for Transmission Case $1$}\label{seccas1}
This case occurs when $g_{12}\leq g_{10}, g_{21}\leq g_{20}$ and it is the same as the classical non-cooperative MAC.
The probability for this case is obtained as follows.
\begin{lem}
The probability for case $1$ is given as
\begin{align}\label{prob1}
P[g_{12}> g_{10}, g_{21}> g_{20}]=\frac{\mu_{10}}{\mu_{12}+\mu_{10}}
\frac{\mu_{20}}{\mu_{21}+\mu_{20}}
\end{align}
where $\mu_{ij}$ is the mean of $g_{ij}^2$ for $i\in\{1,2\}$ and $j\in\{0,1,2\}$.
\end{lem}
\begin{proof}
See Appendix B.
\end{proof}
The common and individual outage probabilities
$(P_{c1},\;P_{11},\;P_{21})$ for this case are defined in
\cite{NaR}. Hence, the outage probabilities for this case are similar to that in \cite{NaR} except that each outage probability is conditioned on the event that $g_{12}\leq g_{10}$ and
$g_{21}\leq g_{20}$.
\subsection{Outage Probability for Transmission Case $2$}
This case applies when $g_{12}> g_{10}, g_{21}> g_{20}$, which allows full
cooperation between the two UEs.
The probability for this case is the same as (\ref{prob1}) but replacing $\mu_{10}$ by $\mu_{12}$ and $\mu_{20}$ by $\mu_{21}$ in the numerator.
In this case, since the two UEs perform rate splitting and partial
decode-forwarding, the target rates $(R_{1},R_{2})$ are split into the
cooperative and private target rates as described in Section \ref{sec:
  simpsch}. Different from the non-cooperative MAC, here outage can occur
at either UE or at the BS. We first analyze outage
probabilities at the UEs and the BS separately, then combine
them to obtain the overall outage probability.
\subsubsection{Outage at the UEs}
As $\text{UE}_1$ has no CSI about $g_{12}$, the transmission rate $R_{12}$ may
exceed $J_1$ in (\ref{spr12}),  which is the maximum rate supported by the
fading channel to $\text{UE}_2$. Therefore, there is a possibility for outage
at $\text{UE}_2$. The outage probability at $\text{UE}_2$ $(P_{m2})$ is given as
\begin{align}\label{outm1m2}
P_{m2}=&P\left[\alpha_2\log(1+g_{12}^2\rho_{11})\leq R_{12}|g_{12}> g_{10}, g_{21}> g_{20}\right]
=P\left[g_{12}^2\leq \frac{2^{\alpha_2R_{12}}-1}{\rho_{11}}|g_{12}> g_{10}\right]
\end{align}
Similar formula holds for the outage probability at $\text{UE}_1$ $(P_{m1})$.
\subsubsection{Outage at the Base Station}
The outage at the BS is considered when there are no outages at
the UEs. This outage is tied directly with the decoding constraints of
the cooperative and private information parts as shown
in (\ref{sprd}). This outage consists of two parts, for the cooperative and
the private information.

Because of the superposition coding structure that each private part is
superimposed on both cooperative parts, an outage for either of the
cooperative information parts leads to an outage for both private
parts. Hence we only need to consider the common outage for the cooperative
parts, but need to consider both the common and individual outage for the
private parts.

\begin{rem}
For the achievable rate region in (\ref{th1Grrsimp}), we look at the combination of (\ref{spr12}) and (\ref{sprd}) and we show in Theorem \ref{cor1} that
two rate constraints $R_1+R_{20}\leq J_6$ and $R_{10}+R_2\leq J_7$ in (\ref{sprd}) are redundant. However, in the outage analysis, we look at the outage at the
UEs and the BS separately. Hence, these $2$ constraints at the BS are active and they affect the outage of the cooperative parts.
\end{rem}
\subsubsection*{Outage of the Cooperative Parts}
From (\ref{sprd}), the rate constraints for the cooperative parts are
\begin{align}\label{ouc2p1}
R_{12}\leq&\; J_6-(R_{10}+R_{20}),\;\;R_{21}\leq\; J_7-(R_{10}+R_{20}),\;\;
R_{12}+R_{21}\leq\; J_8-(R_{10}+R_{20}).
\end{align}
For fixed target rates $(R_{10},R_{12},R_{20}, R_{21})$, a common outage of
the cooperative parts occurs when the cooperative target rate pair
$(R_{12},R_{21})$ lies outside the region obtained from (\ref{ouc2p1}). The
probability of this cooperative common outage is given as
\begin{align}\label{prcc}
\!\!\!\!P_{cc}=1-P\bigg[&R_{12}\leq J_6-(R_{10}+R_{20}),\;R_{21}\leq J_7-(R_{10}+R_{20}),\;R_{12}+R_{21}\leq J_8-(R_{10}+R_{20})|\xi_1\bigg]
\end{align}
where $\xi_1$ is the event that case $2$ happens and there is no outage at the UEs, which is defined as
\begin{align}
\xi_1=\bigg\{&g_{12}>\max\left(\sqrt{\frac{2^{\alpha_2R_{12}}-1}{\rho_{11}}},g_{10}\right),\;\;g_{21}>
\max\left(\sqrt{\frac{2^{\alpha_1R_{21}}-1}{\rho_{22}}},g_{20}\right)\bigg\}
\end{align}
\subsubsection*{Outage of the Private Parts}
For the private parts, the rate constraints obtained from $(\ref{sprd})$ are
\begin{align}\label{prp}
R_{10}\leq J_3,\;R_{20}\leq J_4,\;R_{10}+R_{20}\leq J_5
\end{align}
This region is similar to the classical MAC. Hence, the common $(P_{cp})$
and individual $(P_{1p},P_{2p})$ outage probabilities for private parts
can be obtained as
\begin{align}\label{fornr}
&P_{cp}=P[R_{10}> J_3,\;R_{20}\leq J_5-J_1|\xi_2,\xi_1]
+P[R_{20}> J_4,\;R_{10}\leq J_5-J_2|\xi_2,\xi_1]\\
&\;\;+P[R_{10}\leq J_5-J_2,\;R_{20}> J_5-J_1,\;R_{10}+R_{20}>J_5|\xi_2,\xi_1],\nonumber\\
&P_{1p}=P[R_{10}> J_3,\;R_{20}\leq J_5-J_1|\xi_2,\xi_1]+P[R_{10}\leq J_5-J_2,\;R_{20}> J_5-J_1,\;R_{10}+R_{20}>J_5|\xi_2,\xi_1],\nonumber\\
&P_{2p}=P[R_{20}> J_4,\;R_{10}\leq J_5-J_2|\xi_2,\xi_1]+P[R_{10}\leq J_5-J_2,\;R_{20}> J_5-J_1,\;R_{10}+R_{20}>J_5|\xi_2,\xi_1]\nonumber
\end{align}
where $\xi_2$ is the event that (\ref{ouc2p1}) holds.
\begin{rem}
Although the probabilities in (\ref{fornr}) are in similar form to those
in \cite{NaR}, they are conditional probabilities that depend on the outage event for the common part in (\ref{ouc2p1}). Hence, the formulas in (\ref{fornr}) cannot be evaluated in closed forms as in \cite{NaR}.
\end{rem}
\subsubsection*{Outage at the Base Station}
Since an outage for any cooperative part leads to an outage for both
private information parts, the individual outage at the BS in
(\ref{outc2}) occurs with probability $(P_{b1})$ if the cooperative parts
are in outage or the cooperative parts are decoded correctly but the
private information part of $\text{UE}_1$ is in outage. Similar analysis applies for
$P_{b2}$. The common outage occurs at the BS with probability
$P_{bc}$ if the cooperative parts are in outage or the cooperative parts
are decoded correctly but either or both private parts are in
outage. Hence, we have
\begin{align}\label{pcio}
P_{bc}&=P_{cc}+\bar{P}_{cc}P_{cp},\;\;P_{b1}=P_{cc}+\bar{P}_{cc}P_{1p},\;\;P_{b2}=P_{cc}+\bar{P}_{cc}P_{2p}
\end{align}
where $P_{cc}$ is given in (\ref{prcc}), $\bar{P}_{cc}=1-P_{cc}$ and $P_{cp},\;P_{1p}$ and $P_{2p}$ are given as in (\ref{fornr}).
\subsubsection{Overall Outage for Case $2$}
The outage probability for case $2$ can now be obtained from
(\ref{outm1m2}) and $\eqref{pcio}$ as follows.
Common outage occurs if there is an outage at $\text{UE}_1$, or there is no outage at $\text{UE}_1$ but an outage at $\text{UE}_2$, or there is no
outage at either UE but an outage at the BS. Similar analysis holds for the individual outages. Therefore, the common $(P_{c2})$ and individual
$(P_{12},P_{22})$ outage probabilities become
%
\begin{align}\label{outc2}
P_{c2}=&P_{m1}+\bar{P}_{m1}P_{m2}+\bar{P}_{m1}\bar{P}_{m2}P_{bc},\\
P_{12}=&P_{m1}+\bar{P}_{m1}P_{m2}+\bar{P}_{m1}\bar{P}_{m2}P_{b1},\;\;
P_{22}=P_{m1}+\bar{P}_{m1}P_{m2}+\bar{P}_{m1}\bar{P}_{m2}P_{b2},\nonumber
\end{align}
where $\bar{P}_{m1}=1-P_{m1},$ $\bar{P}_{m2}=1-P_{m2},$ $P_{bc},$ $P_{b1}$ and $P_{b2}$ are the outage probabilities at the BS (\ref{pcio}).
\begin{rem}
Since an outage at either UE will cause an outage of the common information part, and each private information part is superposed on both common parts, UE outages contribute to both the common and private outages overall.
\end{rem}
\subsection{Outage Probability for Transmission Cases $3$ and $4$}\label{seccas34}

This case occurs when $g_{12}> g_{10}, g_{21}\leq
g_{20}$,  which allows one way of cooperation from $\text{UE}_1$ to $\text{UE}_2$. The probability of this case is the same as (\ref{prob1}) but replacing $\mu_{10}$ by $\mu_{12}$ in the numerator.

In this case, only the target rate of $\text{UE}_1$ $(R_{1})$ is divided into
cooperative and private target rates as $R_1=R_{10}+R_{12}$.
The outage probability now depends on the outage probability at
$\text{UE}_2$ and the BS. Since the outage at $\text{UE}_2$ is identical to $P_{m2}$ given in (\ref{outm1m2}), we
only analyze  the outage at the BS for this case.

Similar to Case $2$, the outage at the BS consists of two parts:
cooperative and private outages.
In this case, there is only one cooperative information part with rate
constraint obtained from (\ref{sprd}) as
\begin{align}\label{ouc3p1}
R_{12}\leq&\; J_6-(R_{10}+R_{2}).
\end{align}
Thus, the outage probability for the cooperative part is
\begin{align}\label{ouc3p2}
P_{cr}=P\big[R_{12}> J_6-(R_{10}+R_{2})|\xi_3\big],
\end{align}
where $\xi_3$ is the event that case $3$ happens and there is no outage at $\text{UE}_2$, which is given as
\begin{align}
\xi_3=\left\{g_{12}>\max\left(\sqrt{\frac{2^{\alpha_2R_{12}}-1}{\rho_{11}}},g_{10}\right),\;g_{21}\leq g_{10}\right\}.
\end{align}
For the private parts, the outage probability is similar to Case $2$ but with $\xi_2$ pertains to the event that (\ref{ouc3p1}) holds.  Hence, the common and individual outage probabilities at the BS are given as
\begin{align}\label{pcio3}
P_{bc}&=P_{cr}+\bar{P}_{cr}P_{cp},\;\;
P_{b1}=P_{cr}+\bar{P}_{cr}P_{1p},\;\;
P_{b2}=P_{cr}+\bar{P}_{cr}P_{2p},
\end{align}
where $P_{cp},$ $P_{1p}$ and $P_{2p}$ are given in (\ref{fornr}) with $R_{20}=R_2$ and $\xi_2$ pertains to the event that (\ref{ouc3p1}) holds.
Finally, the overall common $(P_{c3})$ and individual $(P_{13},P_{23})$ outage probabilities for this case are given as
\begin{align}\label{outc3}
P_{c3}=&P_{m2}+\bar{P}_{m2}P_{bc},\;\;
P_{13}=P_{m2}+\bar{P}_{m2}P_{b1},\;\;
P_{23}=P_{m2}+\bar{P}_{m2}P_{b2},
\end{align}
with $P_{bc},$ $P_{b1}$ and $P_{b1}$ as in (\ref{pcio3}).

Case $4$ occurs when $g_{12}\leq g_{10}, g_{21}> g_{20}$
and is simply the opposite of Case $3$.
\subsection{Outage Rate Region}
 The last two subsections provide the formulation and analysis of the
 outage probabilities at a given target rate pair. Some services may
 require target outage probabilities instead of the target rates. For these
 services, we can obtain the individual and common \emph{outage rate
   regions} as follows.
 \begin{mydef} \label{def1}
 For given target outage probabilities $(\beta_1,\beta_2)$, the individual outage rate region of the proposed D$2$D uplink cooperative scheme  consists of all rate pairs $(R_1,R_2)$ such that
 \begin{align}\label{outrr}
 P_1(R_1,R_2,\underline{\rho})&\leq \beta_1,\;\;P_2(R_1,R_2,\underline{\rho})\leq \beta_2
 \end{align}
 where $\underline{\rho}=(\rho_{10},\rho_{20},\rho_{11},\rho_{22},\rho_{13},\rho_{23})$ represents all possible power allocations satisfying the power constraints in (\ref{powcsch2}). $P_1$ and $P_2$ are functions of $(R_1,R_2,\underline{\rho})$ as shown in (\ref{outc2}) and (\ref{outc3}).\\
 Similarly, the common outage rate region consists of all rate pairs $(R_1,R_2)$ such that
 \begin{align}\label{outrr2}
 P_c(R_1,R_2,\underline{\rho})&\leq \min\{\beta_1,\beta_2\}
 \end{align}
with $P_c$ as given in  (\ref{outc2}) and (\ref{outc3}).
 \end{mydef}
\section{Comparison with Frequency Division Schemes}\label{sec: regsoutb}
In this section, we compare the proposed TD scheme with the existing half-duplex schemes  based on FD or CDMA  in \cite{Erkip, Mes7, SErkip}. We show that the proposed scheme achieves the same or better rate region while has simpler transmit signals and significantly shorter decoding delay. Moreover, we formulate the outage probability for the existing schemes as they are unavailable in these prior works.
\subsection{Three-Band Frequency Division}
Based on the original information-theoretic scheme in \cite{fcc2005fof, Erkip}, frequency division can be used in the proposed scheme instead of time division   as proposed. In FD implementation, the bandwidth of each transmission block is divided into $3$ bands and the transmissions in the first $2$ bands are similar to the first $2$ phases in the TD scheme except that both UEs transmit at the same time  (on different frequency bands).
In the $3^{\text{rd}}$ band, both users will transmit concurrently. However, because in the same block of time, the two users are still exchanging current cooperative information on the first $2$ bands, then in the $3^{\text{rd}}$ band, they can only send the previous and not the current cooperative information  as in \cite{fcc2005fof, Erkip}.
Therefore, frequency-division implementation requires block Markov signaling structure which requires backward decoding with long block delay, or sliding window decoding with one block delay. In \cite{SErkip}, a half-duplex cooperative OFDMA system  with N subchannels is proposed where these subchannels are divided into $3$ sets. Considering
these $3$ sets as the $3$ phases of the FD scheme, the transmission and the achievable rate regions in these two schemes are similar.

In comparison, for $3$-band FD and $3$-set OFDMA, the  information dependency  between consecutive blocks complicates the signaling by requiring a block Markov signal structure. The proposed scheme, by using time division, overcomes this block Markov requirement and allows the forwarding of information in the same block.  Moreover, backward or sliding window
decoding is required for FD implementation because of the block Markov structure, which for Gaussian channel leads to the same achievable rate region of the proposed scheme but with at least one block
delay whereas the proposed scheme incurs no block decoding delay. Based on this discussion, we obtain the following corollary:
\begin{cor}\label{cor3b}
The proposed $3$-phase TD scheme achieves the same rate region of the $3$-band FD or the $3$-set OFDMA scheme while having simpler transmit signals  and shorter decoding delay.
\end{cor}

\subsection{Two-Band Frequency Division}
In \cite{Mes7}, another half-duplex scheme is proposed based on FD. In each block, the bandwidth is divided into two bands with widths $\beta$ and $\bar{\beta}=1-\beta$. Each band is divided by half into two sub-bands. In the first band,  $\text{UE}_2$ works as a relay for  $\text{UE}_1$ while the opposite happens in the second band. In the first sub-band,  $\text{UE}_1$ sends its information with $\rho_{12}$ power and  $\text{UE}_2$ decodes it. In the second sub-band,  $\text{UE}_1$ and  $\text{UE}_2$ allocate the powers $\rho_{1}^{(1)}$ and $\rho_{1}^{(2)}$, respectively to send the previous information of $\text{UE}_1$ to the BS. The opposite happens in the second band. The BS employs sliding window decoding. The achievable rate of this scheme consists of the rate pairs $(R_1, R_2)$ satisfying \cite{Mes7}
\begin{align}\label{MesbaS}
R_1&\leq \min\{A_1, A_3\},\;R_2\leq\; \min\{A_2,A_4\},\;\;
A_1=0.5\beta\log(1+g_{12}^2\rho_{12}),\;A_2=0.5\bar{\beta}\log(1+g_{21}^2\rho_{21}),\\
A_3&=0.5\beta\log\left(1+g_{10}^2\rho_{12}+\big(g_{10}\sqrt{\rho_{1}^{(1)}}+g_{20}\sqrt{\rho_{1}^{(2)}}\big)^2\right),\nonumber\\
A_4&=0.5\bar{\beta}\log \left(1+g_{20}^2\rho_{21}+\big(g_{10}\sqrt{\rho_{2}^{(1)}}+g_{20}\sqrt{\rho_{2}^{(2)}}\big)^2\right),\nonumber
\end{align}
for some $0\leq\beta\leq1$ and power allocation satisfying
\begin{align}
\beta(\rho_{12}+\rho_{1}^{(1)})+\bar{\beta}\rho_2^{(1)}&\leq {\cal P}_1,\;\;
\bar{\beta}(\rho_{21}+\rho_{2}^{(2)})+\beta\rho_1^{(2)}\leq {\cal P}_2.
\end{align}
\begin{cor}\label{cor_comp_sc2}
Compared with the proposed scheme, the $2$-band scheme has longer delay and smaller rate region.
\end{cor}
\begin{proof}
The $2$-band scheme has one block delay because of using sliding window decoding. Moreover, the scheme uses neither  information splitting nor superposition coding. These two techniques, which are employed in the proposed scheme, enlarge the rate region  as shown in Appendix A.
\end{proof}
\subsection{Outage Probability Analysis}
Next, we derive the outage probability for the existing  schemes in \cite{Erkip, Mes7, SErkip} as outage results are unavailable in these previous works.
\subsubsection{Outage for the $3$-band Frequency-Division Transmission}
For the $3$-band FD scheme derived from \cite{Erkip} and the OFDMA scheme in \cite{SErkip}, the outage probability is given as follows.
\begin{cor}\label{corTsh}
The outage probability for the $3$-band FD or OFDMA scheme  is similar to the proposed TD scheme except that the cooperative common outage for Case $2$ in (\ref{prcc}) is replaced with
\begin{align}\label{prFDcc}
P_{cc}=1-P\bigg[R_{12}+R_{21}\leq\; J_8-(R_{10}+R_{20})|\xi_1\bigg]
\end{align}
\end{cor}
\begin{proof}
Since the BS in both schemes employs backward decoding, the rate constraints at the BS are similar to (\ref{sprd}) but without $R_1+R_{20}\leq J_6$
and $R_{10}+R_2\leq J_7$. Hence, these $2$ constraints are removed from the cooperative common outage in (\ref{prFDcc}).
\end{proof}
\subsubsection{Outage for the $2$-band Frequency-Division Transmission}
For this scheme, the outage probability can be formulated considering the achievable rate region in (\ref{MesbaS})
and following similar procedure for the outage of the proposed scheme where in case
\begin{itemize}
  \item Case $1$, direct transmission is used with $\rho_{1}^{(2)}=\rho_2^{(1)}=0$.
  \item Case $2$, cooperation from both UEs.
  \item Case $3$, cooperation from $\text{UE}_1$ and direct transmission form  $\text{UE}_2$ with $\rho_{2}^{(1)}=0$.
  \item Case $4$, cooperation from  $\text{UE}_2$ and direct transmission form  $\text{UE}_1$ with $\rho_{1}^{(2)}=0$.
\end{itemize}
Then, the outage probability is given as follows.
\begin{cor}
For the $2$-band FD scheme with the achievable rate region in (\ref{MesbaS}), the common outage probability is given as in (\ref{outtc}) but with
\begin{align}
P_{c1}&=P\big[R_1>A_3,R_3\leq A_4\big]\big]+P\big[R_1>A_3,R_3>A_4\big]
\;+P\big[R_1\leq A_3,R_3>A_4\big],\\
P_{c2}&=P_{m1}+\bar{P}_{m1}P_{m2}+\bar{P}_{m1}\bar{P}_{m2}P_{c1},\;\;
P_{c3}=P_{m1}+\bar{P}_{m1}P_{c1},\;P_{c4}=P_{m2}+\bar{P}_{m2}P_{c1}.\nonumber
\end{align}
The individual outage probabilities are formulated similarly. However, since in each bands, one UE works as a relay for the other UE, the outage at one UE will lead to an outage of the other UE information at the BS and not both information as in the proposed scheme. Hence, for  $\text{UE}_1$ outage, we have
\begin{align}
P_{11}&=P\big[R_1>A_3,R_3\leq A_4\big]\big]+P\big[R_1>A_3,R_3>A_4\big],\nonumber\\
P_{12}&=P_{m2}+\bar{P}_{m2}P_{11},\;P_{13}=P_{12}, P_{14}=P_{11}.
\end{align}
The outage at  $\text{UE}_2$ is formulated similarly.
\end{cor}
\begin{proof}
Obtained following similar procedure of the proposed scheme.
\end{proof}
\vspace{-3mm}
\subsection{Tradeoff between decoding delay and rate constraints}
Comparison between the proposed TD scheme and the $3$-band FD and OFDMA schemes in \cite{SErkip} reveals the following interesting trade-offs among  decoding delay, rate constraints  and outage performance.
 Based on formula (\ref{sprd}) and the proof of Corollary \ref{corTsh}, the BS can decode with fewer rate constraints if it is allowed longer decoding delay, as it  can use more received signals in order
 to have better estimation of the transmitted information. Specifically, the OFMA scheme in \cite{SErkip} employs backward decoding where the BS, at each block, decodes the current private information and the previous cooperative information given that it knows the current cooperative information. This knowledge reduces the error events and hence, the rate constraints stemmed from the decoding at the BS. Therefore rate constraints $J_6$ and $J_7$ are not present in the FD backward decoding implementation but are present in our proposed TD scheme.

For application in wireless channels, however, this difference in decoding rate constraints do not matter for the overall achievable rate region, as shown in Corollary \ref{cor1} and Corollary \ref{cor3b}. This equivalency happens since after we combine the
 rate constraints from the UEs and the BS, we then find as  in Corollary \ref{cor1} that the additional constraints, stemmed from the decoding at the BS, are redundant. For outage performance, however, all decoding rate constraints matter as we need to  separately consider the outages at UEs and
  the BS. The difference in decoding rate constraints then leads to different outage formulas as shown in (\ref{prcc}) and (\ref{prFDcc}). Nevertheless, numerical results in Section \ref{sec:Num}
 show that both the $3$-band FD scheme and our $3$-phase TD scheme have the same outage performance, which suggests that the additional rate constraints in our TD scheme do not matter in outage performance either.

\section{Numerical Results}\label{sec:Num}
In this section, we provide numerical results and analysis for the achievable rate
region, outer bound, outage probabilities and outage rate region derived in Sections \ref{sec: simpsch} and \ref{sec: outana}. In these simulations, all the links are Rayleigh fading
channels with parameters specified in each simulation. All the simulations are obtained using $10^5$ samples for each fading channel.
Figure  \ref{fig:asyro} shows the ergodic achievable rate
regions while Figures \ref{fig:OUT31}---\ref{fig:OUT4} show the outage performance.
\subsection{Achievable Rate Region}
Figure \ref{fig:asyro} is obtained with ${\cal P}_1={\cal P}_2=2$ with all possible power allocations and phase durations.
It compares between the achievable rate region of the proposed and existing transmissions and the
outer bound for asymmetric channels. Results are plotted  with $\mu_{10}=4$ and $\mu_{20}=1$ while $\mu_{12}=\mu_{21}=16$.

As
discussed in Section \ref{sec:out. cap}, results show that the achievable rate
region of the proposed TD cooperative transmission is close to the outer bound since the ratios $\mu_{12}/\mu_{10}$
and $\mu_{21}/\mu_{20}$ are high. Moreover, $\text{UE}_2$ that has weaker link to the BS obtains higher gain in the individual rate than
$\text{UE}_1$. This is because $\text{UE}_1$ has stronger link to the BS and can work as a relay for $\text{UE}_2$ information. Comparing with other transmissions, the $3$-band FD scheme in \cite{Erkip, SErkip} achieves the same rate region of the proposed scheme as mentioned in Corollary \ref{corTsh}. However, the proposed $3$-phase TD cooperative transmission outperforms the resource partitioning (RP) with orthogonal transmission (LTE-A) and the concurrent transmission with SIC as no cooperation
is employed in these $2$ schemes. The proposed scheme also outperforms the $2$-band frequency division FD in \cite{Mes7} since neither rate
splitting nor superposition coding is employed in $2$-band cooperative scheme as shown in Corollary \ref{cor_comp_sc2}.
\noindent
\begin{figure}[t]
    \centering
    \begin{minipage}[b]{0.48\linewidth}
    \includegraphics[width=0.90\textwidth, height=60mm]{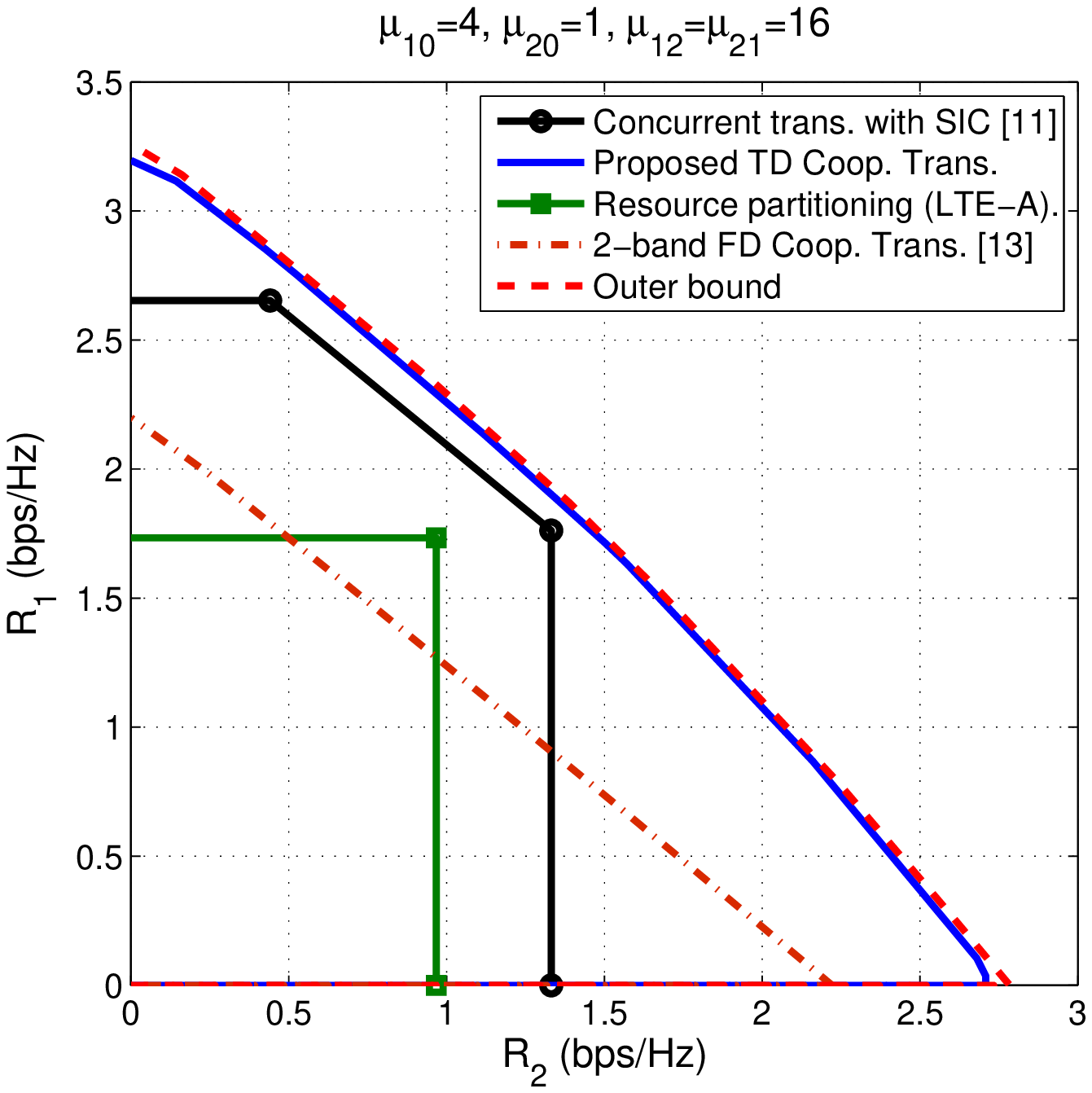}
    \caption{Achievable rate regions and outer bounds for asymmetric half-duplex D$2$D cooperative transmissions (TD=time-division, FD=frequency-division and RP= resource partitioning).}
    \label{fig:asyro}
    \end{minipage}
    \hfill
\begin{minipage}[b]{0.48\linewidth}
    \includegraphics[width=0.90\textwidth, height=60mm]{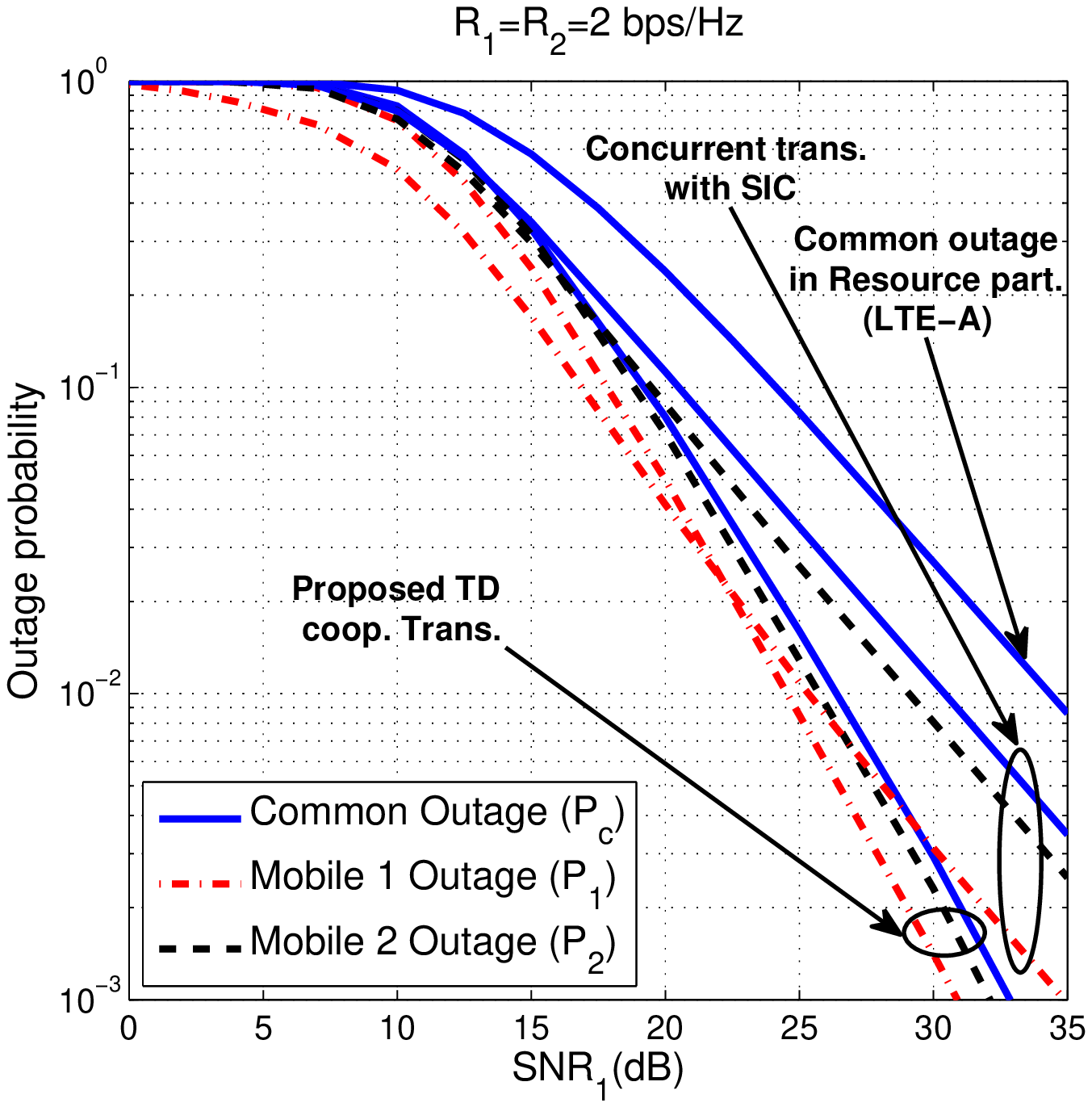}
    \caption{Comparison between the proposed and existing schemes in terms of common and individual outage probabilities $\text{SNR}_1$ with $R_1=R_2=2$.}
    \label{fig:OUT31}
    \end{minipage}
\vspace*{-4mm}
\end{figure}
\noindent
\begin{figure}[!t]
    \centering
    \begin{minipage}[b]{0.48\linewidth}
    \includegraphics[width=0.90\textwidth, height=60mm]{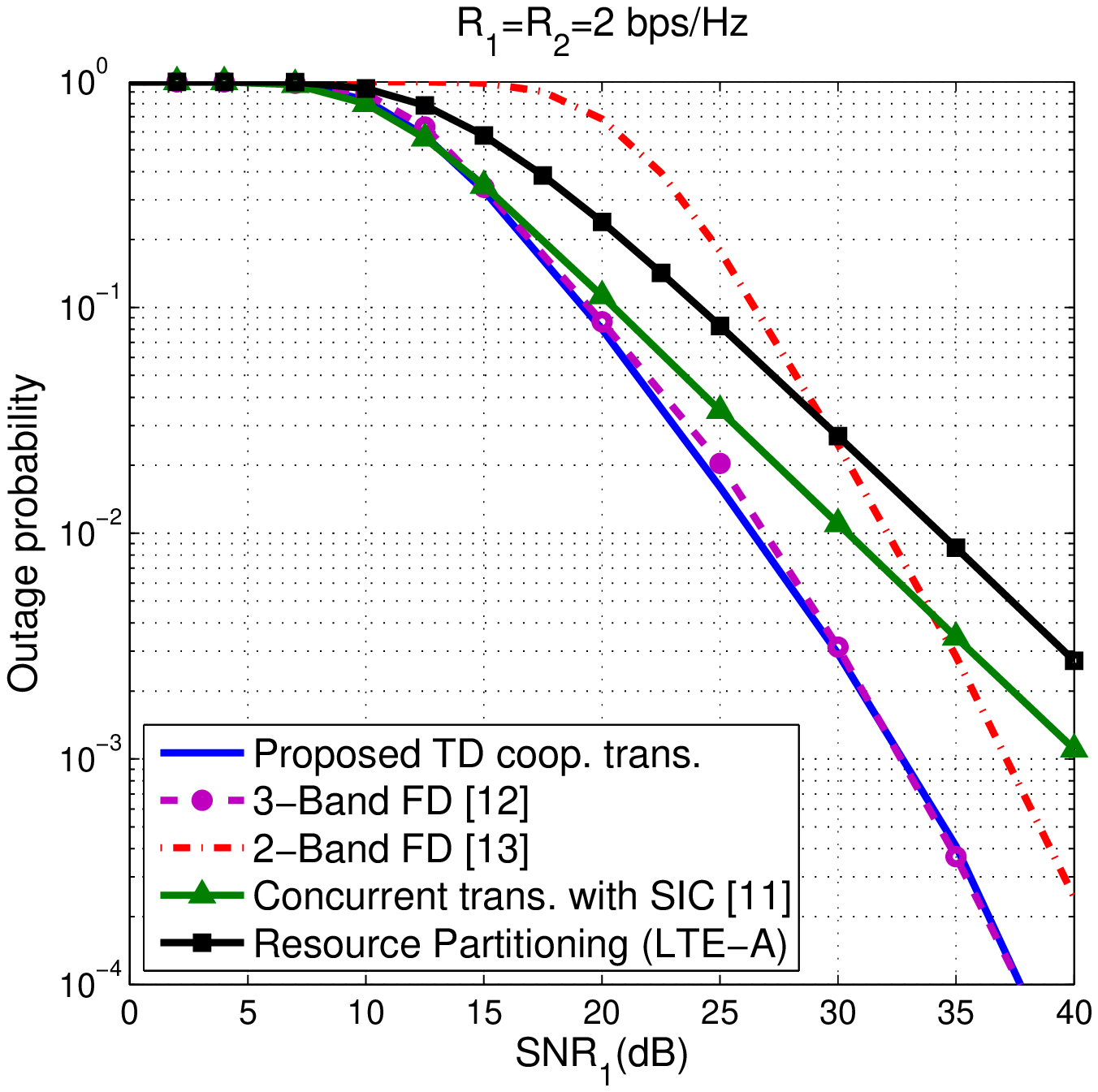}
    \caption{Comparison between the proposed and existing schemes in terms of common outage probabilities $\text{SNR}_1$ with $R_1=R_2=2$.}
    \label{fig:OUT331}
    \end{minipage}
    \hfill
\begin{minipage}[b]{0.48\linewidth}
    \includegraphics[width=0.90\textwidth, height=60mm]{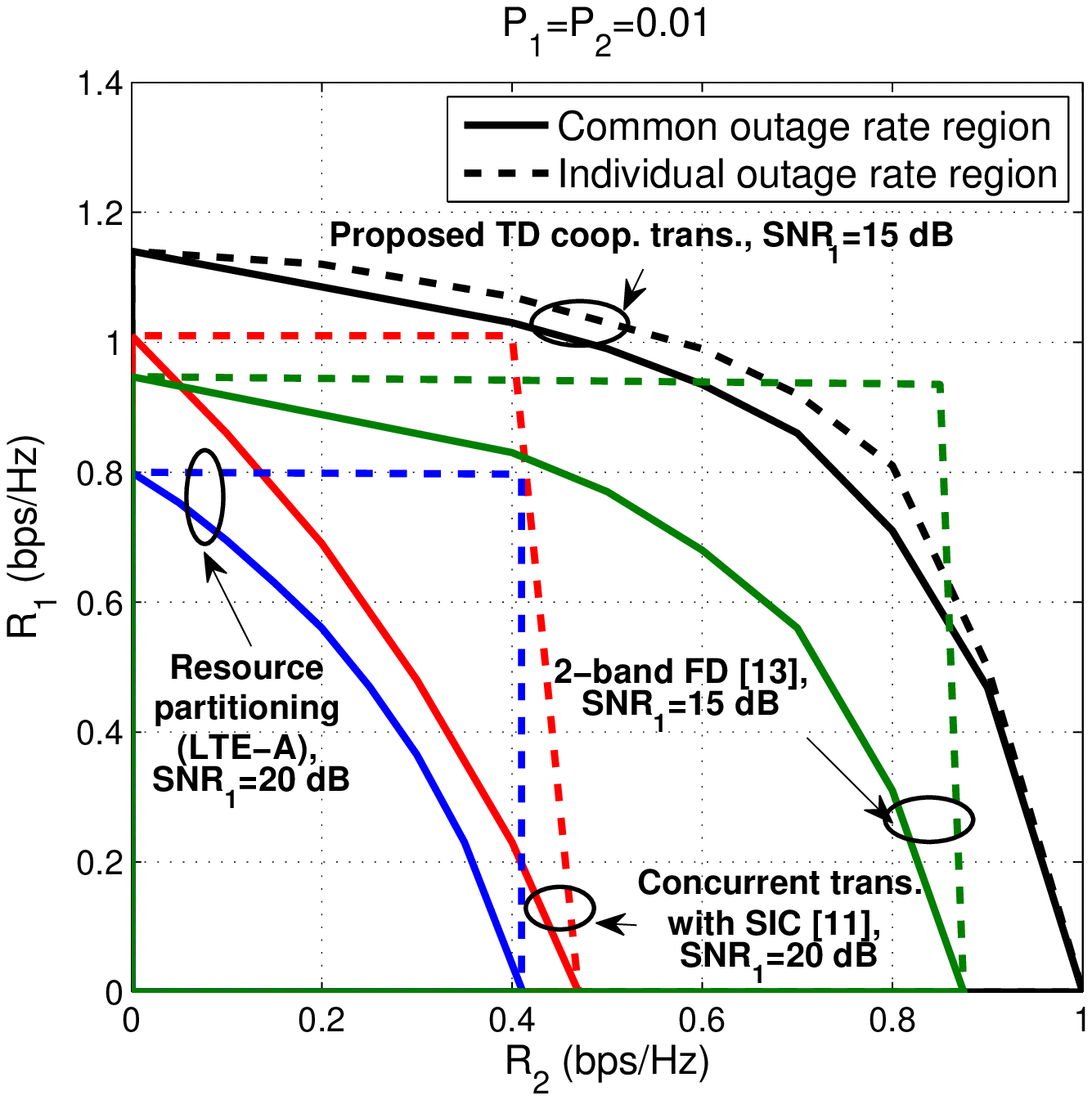}
    \caption{Common and individual outage rate regions for the proposed and existing uplink transmissions at $P_1=P_2=1\%$ and different $\text{SNR}_1$.}
    \label{fig:OUT4}
    \end{minipage}
\vspace*{-6mm}
\end{figure}
\subsection{Outage Probability}
We now provide numerical results for the formulated outage probabilities and outage rate region. The simulation settings and channel configuration are: $d_{10}=20,$ $d_{20}=30,$ $d_{12}=d_{21}=12,$ ${\cal P}_1={\cal P}_2=\rho$ and $\gamma=2.4$. The average channel gain for each link is given as
$\mu_{ij}=\frac{1}{d{ij}^{\gamma}}$. With these settings, we define the average received $\text{SNR}$ at the BS for signals from $\text{UE}_1$ ($\text{SNR}_1$) and $\text{UE}_2$ $(\text{SNR}_2)$ as follows.
\begin{align}
\text{SNR}_1&=10\log\left(\frac{\mu_{10}\rho}{d_{10}^{\gamma}}\right),\;\;
\text{SNR}_2=10\log\left(\frac{\mu_{20}\rho}{d_{20}^{\gamma}}\right)=\text{SNR}_1
+10\log\left(\frac{\mu_{20}d_{10}^{\gamma}}{\mu_{10}d_{20}^{\gamma}}\right)
\end{align}
The phase durations are set as follows: case $2$: $\alpha_1=\alpha_2=0.25$ and $\alpha_3=0.5$, case $3$: $\alpha_1=0.4$ and $\alpha_3=0.6$ and case $4$:  $\alpha_2=0.4$ and $\alpha_3=0.6$. We fix  the phase duration
to simplify computation while the optimal power allocations and rate splitting are obtained numerically. For the $2$-band FD
scheme in (\ref{MesbaS}) \cite{Mes7}, we fix the $2$ bandwidths to half $(\beta=\bar{\beta}=0.5)$ for all cases.

Figure \ref{fig:OUT31} shows the outage probabilities versus $\text{SNR}_1$ for the proposed scheme and concurrent transmission with SIC. Results confirm our expectation that the common outage probability
is higher than individual outage probabilities. Moreover, at equal transmission rates, the individual outage probability for $\text{UE}_2$ is higher than $\text{UE}_1$ since $\text{UE}_2$ has weaker direct link. For low $\text{SNR}$, the non-cooperative concurrent transmission with SIC scheme  has lower outage probability than cooperative
scheme, but the outage in this range is too high for practical interest (above $10\%$).
As $\text{SNR}$ increases, the cooperative scheme starts outperforming the non-cooperative scheme. This happens because in the
 cooperative scheme, each UE transmits over a fraction of time instead of the whole time as in the concurrent transmission with SIC.
 Moreover, both UEs use part of their power to exchange information such that they transmit coherently in
 the $3^{\text{rd}}$ phase. At low $\text{SNR}$, the coherent transmission has lower effect compared with the power loss in exchanging information; hence, the concurrent transmission with SIC outperforms the cooperative scheme. As $\text{SNR}$ increases, however,
 the gain obtained from coherent transmission becomes dominant such that the cooperative scheme outperforms the
 concurrent transmission with SIC. These results are obtained with arbitrary fixed phase durations; thus if they are optimally chosen, the
 cooperative scheme will outperform the non-cooperative scheme at an even lower $\text{SNR}$.

 Figure \ref{fig:OUT31} also shows that the diversity order of the cooperative scheme is $2$. This result is in contrary to that in \cite{Lantse} which shows that decode-forward scheme for half-duplex relay channel achieves a diversity order of $1$ only. The difference comes from the fact that in \cite{Lantse}, the source only transmits in the $1^{\text{st}}$ phase and there is no coherent transmission in the $2^{\text{nd}}$ phase while
   the relay always decodes even when its link with the source is weak. However, in our scheme, there is a coherent transmission in $2^{\text{nd}}$ phase and each UE only decodes if the cooperative link
  is stronger than the direct link. Intuitively, our scheme always requires $2$ links to be weak in order to lose the information of
  any UE \cite{Lantse}.   Consider the information of $\text{UE}_1$,  if the cooperative link is weak, this information will be lost if the direct
  link of $\text{UE}_1$  is also weak. If the cooperative link is strong, this information will be lost if both direct links are weak.

 Figure \ref{fig:OUT331} compares the common outage probabilities of the proposed transmission, $3$-band FD \cite{SErkip},
 $2$-band FD \cite{Mes7}, concurrent with SIC \cite{hansen1977mfr} and RP with orthogonal transmission (current LTE-A). Results show that  compared with non-cooperative transmissions,
 cooperation improves the diversity whether using the proposed TD transmission, $3$ or
 $2$-band FD. However, $2$-band FD transmission requires more power to start outperforming the non-cooperative transmissions. The proposed TD scheme and $3$-band FD transmissions
  have the same outage performance although their outage formulas are slightly different as shown in Corollary \ref{cor3b}.
\subsection{Outage Rate Region}
Figure \ref{fig:OUT4} shows the common and individual outage rate regions when the target outage probability for both UEs is $1\%$. While the target outages are the same, the rate regions are asymmetric because the direct links are different. Results show that the proposed cooperative transmission  has larger regions than the non-cooperative ones (resource partitioning (RP) or concurrent transmission with SIC) even when it has a lower transmit power lower by $5\text{dB}$. Note that considering the outage performance in Figure \ref{fig:OUT31},
the gap between the cooperative and non-cooperative transmissions will increase if the target outage probability decreases, and vice versa.

The $3$-band $\text{UE}_1$ \cite{Erkip, SErkip} has the same performance as the proposed TD transmission. For the $2$-band FD transmission \cite{Mes7}, while the common outage region is always included in that of the proposed transmission, the individual outage region unexpectedly intersects with that of the proposed transmission. This intersection may occur since we fix the phases of the proposed transmissions and the bandwidths for $2$-band transmission in \cite{Mes7}. While fixing them simplifies the computations, these selections can be suboptimal and lead to unexpected results.
The largest regions should be obtained from all possible phase durations and bandwidths.
\section{Conclusion}\label{sec:conclusion}
We have analyzed both the instantaneous achievable rate region and the
outage probability of a D$2$D time-division  cooperative scheme in uplink
cellular communication. The scheme employs rate splitting, superposition
coding, partial decode-forward relaying and ML decoding in a $3$-phase
half-duplex transmission.
When applied to fading channels, outage probabilities can be
computed based on outages at the user equipments and the base station. We formulate for the first
time both the common and individual outage probabilities for a cooperative transmission scheme. Moreover, we formulate the outage performance of the existing FD-based schemes and compare them with  the proposed scheme. Numerical results show significant improvement in the
instantaneous achievable rate region at all SNR, and in the outage performance
 as the SNR increases. These results suggest the use of cooperation for most practical ranges of SNR, except very low SNR where non-cooperative
 schemes may have better outage performance. For future work, it is of interest to apply of the proposed scheme in a multi-cell multi-tier system and analyze its impact on the spectral efficiency and
 outage performance of the entire cellular network.
\appendices
\section{Proof of the Achievable Region in Theorem \ref{nathr1}}
First,  the transmission rates are related to the size of the information sets $({\cal I },$ ${\cal J },$ ${\cal K }$ and ${\cal L})$ as follows.
\begin{align}
R_{12}&=\frac{1}{n}\log|{\cal I }|,\;R_{10}=\frac{1}{n}\log|{\cal J }|,\;\;\;
R_{21}=\frac{1}{n}\log|{\cal K }|,\;R_{20}=\frac{1}{n}\log|{\cal L }|.
\end{align}

The error events at UEs can be analyzed as in
\cite{haykin2005crb, abramowitz1972hmf}. To make these error probabilities
approach zero, $R_{12}$ and $R_{21}$ must satisfy the first two constraints
involving $I_1$ and $I_2$ in (\ref{spr12}).

For the decoding at the BS, the maximum rate $J_{\star}$  achievable for the given channel realization in (\ref{sprd})
is obtained by upper bounding an error event resulted from the decoding
rule in (\ref{decr2}) as follows. Assuming all information vectors are
equally likely, the error probability does not depend on which vector
$(i,k,j,l)$ was sent. Without loss of generality, assume that the event
$E_0 = \{i=k=j=k=1 \text{ was sent}\}$ occurred. Then, $J_{\star}$ ensures
that the probability of error event ${E_{\star}}$ approaches zero as
the transmit sequence lengths increase, where for
\begin{itemize}
    \item $J_3$: ${E_1}=\{\hat{i}=\hat{k}=\hat{l}=1, \hat{j}\neq1\}$, only
      the private part of $\text{UE}_1$ is decoded incorrectly.
    \item $J_4$: ${E_2}=\{\hat{i}=\hat{k}=\hat{j}=1, \hat{l}\neq1 \}$, only
      the private part of $\text{UE}_2$ is decoded incorrectly.
    \item $J_5$: ${E_3}=\{\hat{i}=\hat{k}=1, (\hat{j}, \hat{l})\neq1 \}$,
      both private parts are decoded incorrectly.
    \item $J_6$: ${E_4}= \{\hat{k}=1, (\hat{i}, \hat{j}, \hat{l})\neq1\}$,
      both private parts of two UEs and the cooperative part of $\text{UE}_1$ are decoded incorrectly.
    \item $J_7$: ${E_5}= \{\hat{i}=1, (\hat{k}, \hat{j}, \hat{l})\neq1\}$,
      similar to $J_6$ but with cooperative part of $\text{UE}_2$ decoded incorrectly.
    \item $J_8$: $P_{E_6}= \{(\hat{i}, \hat{k}, \hat{j}, \hat{l})\neq1\}$,
      all information parts are decoded incorrectly.
  \end{itemize}
To analyze the upper bounds for the probabilities of these error events, we will first divide them into $2$ groups where error probabilities in each group have
 similar analysis.
\noindent
\begin{itemize}
\item
The $1^\text{st}$ group contains: $(P_{E_1}, P_{E_2}, P_{E_3})$ and
the $2^\text{nd}$ group contains: $(P_{E_4}, P_{E_5}, P_{E_6})$
\end{itemize}
Since the error analysis for the second group is more complicated, we only analyze the $4^\text{th}$ error event. The error analysis for
the first group can be obtained similarly.

Define $\vartheta_2$ as the event that (\ref{Eq:erst}) holds.
\begin{align}\label{Eq:erst}
P(y|x_1(j,i,1),x_2(l,i,1))\geq&\; P(y|\tilde{x}_1,\tilde{x}_2) \nonumber\\
\leftrightarrow P(y_1|x_{11,i})P(y_2|x_{22,1})P(y_3|x_{13,(j,i,1)},x_{23,(l,i,1)})\geq&\; P(y_1|\tilde{x}_{11})P(y_2|\tilde{x}_{22})P(y_3|\tilde{x}_{13},\tilde{x}_{23}) \nonumber \\
\leftrightarrow P(y_1|x_{11,m_{12}})P(y_3|x_{13,(j,i,1)},x_{20,(l,i,1)})\geq&\; P(y_1|\tilde{x}_{11})P(y_3|\tilde{x}_{13},\tilde{x}_{23})
\end{align}
 Then, the probability of this event is
\begin{align*}
P(\vartheta_2)=&\sum_{x_{11},u_3,x_{13},x_{23}}
P(x_{11,i})p(u_{3,(i,1)})P(x_{13,(j,i,1)}|u_{3,(i,1)})P(x_{23,(l,i,1)}|u_{3,(i,1)})
\end{align*}
 This probability can be bounded as follows \cite{haykin2005crb, abramowitz1972hmf}.
\begin{align*}
P(\vartheta_2)\leq&\sum_{x_{11},u_3,x_{13},x_{23}}
P(x_{11,i})P(x_{13,(j,i,1)})P(x_{23,(l,i,1)})
\left(\frac{P(y_1|x_{11,i})}{P(y_1|x_{11}(1))}\right)^s
\left(\frac{P(y_3|x_{13,(j,i,1)},x_{23,(l,i,1)})}{P(y_3|x_{13}(1,1,1),x_{23}(1,1,1))}\right)^s
\end{align*}
for any $s>0$. Now, let $\vartheta$ be the event that (\ref{Eq:erst}) holds for some $m_{12}\neq1$ and any $j,$ and $l$. Then for any
$0\leq\rho\leq1$, the probability of the event $\vartheta$ can be expressed as follows \cite{haykin2005crb, abramowitz1972hmf}.
\begin{align}\label{Eq:akhr}
P(\vartheta)&\leq\left(\sum_{\boldsymbol{i},\boldsymbol{j},\boldsymbol{l}}P(\vartheta_2)\right)^\rho  \nonumber \\
&=(|{\cal I}|-1)^\rho |{\cal J}|^\rho |{\cal L}|^\rho\left[\sum_{x_{11},u_3,x_{13},x_{23}}P(x_{11})P(x_{13})P(x_{23})
\left(\frac{P(y_1|x_{10})}{P(y_1|\tilde{x}_{10})}\right)^s\left(\frac{P(y_3|x_{13},x_{23})}{P(y_3|\tilde{x}_{13},\tilde{x}_{23})}\right)^s\right]^\rho
\end{align}
where $\boldsymbol{i}\in\{2,...,|{\cal I}|\}$, $\boldsymbol{j}\in\{1,...,|{\cal J}|\}$ and $\boldsymbol{l}\in\{1,...,|{\cal L}|\}$. Then, the probability of interest, $P_{E_4},$ has an upper bound:
\noindent
\begin{align*}
P_{E_4}\leq& \sum_{y_{13},x_{11},x_{13},x_{23}}P(y_1|\tilde{x}_{11})
P(y_3|\tilde{x}_{13},\tilde{x}_{23})P(\tilde{x}_{11})P(\tilde{x}_{13})
P(\tilde{x}_{23})P(\vartheta)
\end{align*}
 By combining the last two equations and by choosing $s=1/(1+\rho)$, $P_{E_4}$ can be written as
\begin{align*}
P_{E_4}\leq&(|{\cal I}|-1)^\rho |{\cal J}|^\rho |{\cal L}|^\rho
\sum_{y_{13}}\left[\sum_{x_{11}}P(\tilde{x}_{11})(P(y_1|\tilde{x}_{11}))^{\frac{1}{1+\rho}}\right]^{1+\rho}
\left[\sum_{x_{13},x_{23}}P(\tilde{x}_{13},\tilde{x}_{23})(P(y_3|\tilde{x}_{13},\tilde{x}_{23}))^{\frac{1}{1+\rho}}\right]^{1+\rho}
\end{align*}
Since the channel is memoryless, $P_{E_4}$ can be expanded as follows.
\begin{align}\label{Eq:errd}
P_{E_4}\leq\; &(|{\cal I}|-1)^\rho |{\cal J}|^\rho |{\cal L}|^\rho
\sum_{y_{1,1}^{\alpha_1n},y_{3,\alpha_d}^{n}}
\left[\sum_{x_{11,1}^{\alpha_1n}}\prod_{t=1}^{\alpha_1 n}P(x_{11_t})(P(y_{1_t}|x_{11_t}))^{\frac{1}{1+\rho}}\right]^{1+\rho} \nonumber \\
&\left[\sum_{x_{13,\alpha_d}^{n},x_{23,\alpha_d}^n}\prod_{i=\alpha_d}^{n}
P(x_{13_t},x_{23_t})(P(y_{3_t}|x_{13_t},x_{23_t}))^{\frac{1}{1+\rho}}\right]^{1+\rho}
\end{align}
Then, by interchanging the order of the products and the summations, (\ref{Eq:errd}) can be simplified to:
\noindent
\begin{align*}
P_{E_4}\leq & (|{\cal I}|-1)^\rho |{\cal J}|^\rho |{\cal L}|^\rho
\prod_{t=1}^{\alpha_1 n}\sum_{y_{1_t}}\left[\sum_{x_{11_t}}P(x_{11_t})(P(y_{1_t}|x_{11_t}))^{\frac{1}{1+\rho}}\right]^{1+\rho}\\
&\prod_{t=\alpha_d}^{n}\sum_{y_{3_t}}\left[\sum_{x_{13_t},x_{23_t}}P(x_{13_t},x_{23_t})
(P(y_{3_t}|x_{13_t},x_{23_t}))^{\frac{1}{1+\rho}}\right]^{1+\rho}
\end{align*}
 Now, since the summations are taken over the inputs and the output alphabets, $P_{E_4}$ can be expressed as follows.
\begin{align}\label{Eq:erfth}
P_{E_4}\leq\;& (|{\cal I}|-1)^\rho |{\cal J}|^\rho |{\cal L}|^\rho
\left(\sum_{y_{1_t}}\left[\sum_{x_{11_t}}P(x_{11_t})(P(y_{1_t}|x_{11_t}))^{\frac{1}{1+\rho}}\right]^{1+\rho}\right)^{\alpha_1 n} \nonumber \\
&\left(\sum_{y_{3_t}}\left[\sum_{x_{13_t},x_{23_t}}P(x_{13_t},x_{23_t})
(P(y_{3_t}|x_{13_t},x_{23_t}))^{\frac{1}{1+\rho}}\right]^{1+\rho}\right)^{\alpha_3n}
\end{align}
Following \cite{haykin2005crb}, $({\cal I}-1){\cal J}{\cal L}$, has the following upper bound:
\begin{align*}
(|{\cal I}|-1)|{\cal J}||{\cal L}|
< 2^{n\left(R_{12}+R_{10}+R_{20}+\frac{2^{-n(R_{10}+R_{20})}}{(\text{ln}2)n}\right)}
\end{align*}
Finally, the bound of $P_{E_4}$ can be expressed as follows.
\begin{align*}
&P_{E_4}\leq 2^{-n\left[\Psi(\rho,P_{16})-\rho(R_{12}+R_{10}+R_{13}+R_{23})\right]},\;\;\text{where},\\
&\Psi(\rho,P_4)=-\left(\alpha_1 \text{log}(q_1)+\alpha_3\text{log}(q_2)+\frac{2^{-n(R_{10}+R_{20})}\rho}{(\text{ln}2)n}\right)\\
&q_1=\sum_{y_{1_t}}\left[\sum_{x_{11_t}}P(x_{11_t})(P(y_{1_t}|x_{11_t}))^{\frac{1}{1+\rho}}\right]^{1+\rho},
q_2=\sum_{y_{3_t}}
\left[\sum_{x_{13_t},x_{23_t}}P(x_{13_t},x_{23_t})(P(y_{3_t}|x_{13_t},x_{23_t}))^{\frac{1}{1+\rho}}\right]^{1+\rho}
\end{align*}
 Now, it can be easily verified that $\Psi(\rho,P_4)|_{\rho=0}=0$. Also, it can be shown that:
\begin{align*}
\left.\frac{d \Psi(\rho,P_4)}{d\rho}\right|_{\rho=0}=&\;\alpha_1 I(X_{11};Y_1)+\alpha_3I(X_{13},X_{23};Y_3)-\frac{2^{-n(R_{10}+R_{20})}}{(\text{ln}2)n}
\end{align*}
Hence, it can be easily noted that $P_{E_4}\rightarrow 0$ as $n\rightarrow \infty$ if:
\begin{align}\label{m73}
\!\!\!\!R_{12}+R_{10}+R_{20}\leq&\alpha_1 I(X_{11};Y_1)+\alpha_3 I(X_{13},X_{23};Y_3)
\end{align}
The other error events in this group can be analyzed similarly. Hence, $(P_{E_5},P_{E_6})\rightarrow 0$ as
$n\rightarrow \infty$ if:
\begin{align}\label{m74}
R_{21}+R_{10}+R_{20}\leq& \alpha_2 I(X_{22};Y_2)+\alpha_3I(X_{13},X_{23};Y_3)\\
R_{12}+R_{21}+R_{10}+R_{20}\leq& \alpha_1 I(X_{11};Y_1)+\alpha_2 I(X_{22};Y_2)
+\alpha_3I(X_{13},X_{23};Y_3)\nonumber
\end{align}
The rate constraints obtained from the first error group are
\begin{align}\label{m71}
R_{10}&\leq \alpha_3 I(X_{13};Y_3|U_3,X_{23}),\;\;
R_{20}\leq \alpha_3 I(X_{23};Y_3|U_3,X_{13}),\;\;
R_{10}+R_{20}&\leq\alpha_3I(X_{13},X_{23};Y_3|U_3).
\end{align}
Finally, by applying the rate constraints in (\ref{m74}) and (\ref{m71}) into the Gaussian channel in (\ref{recsig}) with the signaling in (\ref{sigtr}), we obtain the achievable rate
region given in Theorem \ref{nathr1}. Then, we take the expectation to incorporate the randomness of the fading channel.
\section{Proof of the Occurrence Probability of Case $1$ \ref{nathr1}}
The probability of case $1$ can be simply derived as follows. First since the Rayleigh distribution extends from  $0$ to $\infty$ and the channel parameters are independent, the occurrence probability can be expressed as
\begin{align}
P[g_{12}\leq g_{10},\;g_{21}\leq g_{20}]&=P[g_{12}^2\leq g_{10}^2,\;g_{21}^2\leq g_{20}^2]=P[g_{12}^2\leq g_{10}^2]P[g_{21}^2\leq g_{20}^2]
\end{align}
Then,
\begin{align}
P[g_{12}^2\leq g_{10}^2]&=P[g_{12}^2\leq \mu,\; g_{10}^2=\mu]=\int_0^{\infty} P(g_{12}^2\leq \mu) P(g_{10}^2=\mu) \partial \mu\nonumber\\
&=\int_0^{\infty} (1-\exp\{\frac{-\mu}{\overline{g_{12}^2}}\})\cdot \frac{1}{\overline{g_{10}^2}}\exp\{\frac{-\mu}{\overline{g_{10}^2}}\} \partial \mu
=\frac{\mu_{12}}{\mu_{10}+\mu_{12}}
\end{align}
$P[g_{21}^2\leq g_{20}^2]$ can be obtained in a similar way. Then, $P[g_{12}\leq g_{10},\;g_{21}\leq g_{20}]$ is given as in (\ref{prob1}).
\bibliographystyle{IEEEtran}
\bibliography{references}
\end{document}